\documentclass[10pt,twocolumn,twoside]{IEEEtran}

\newcommand{\black}{\color{black}}

\newcommand{\cp}{{\mc R}}
\usepackage{cite}\usepackage{hyperref}
\usepackage{amsmath,amssymb,amsfonts}
\usepackage{algorithmic}
\usepackage{graphicx}
\usepackage{textcomp}
\usepackage{amsmath,amssymb,bm,bbm,mathrsfs,amscd}
\usepackage{calc}
\usepackage{color}
\usepackage{amsthm}
\usepackage{dsfont}
\usepackage{graphicx}
\usepackage{epstopdf}
\usepackage{epsfig}
\usepackage{tikz}
\usepackage{comment}
\usetikzlibrary{patterns}
\usepackage{amsmath}
\usepackage{amsfonts}
\usepackage{amssymb}
\usepackage{rotating}
\usepackage{mathtools}
\usepackage{color}
\usepackage{subfig}
\usepackage{enumerate,pgfplots}
\newtheorem{theorem}{Theorem}
\newtheorem{definition}{Definition}
\newtheorem{proposition}{Proposition}
\newtheorem{lemma}{Lemma}
\newtheorem{corollary}{Corollary}
\newtheorem{example}{Example}
\newtheorem{remark}{Remark}

\newcommand{\ba}{\begin{array}}
\newcommand{\ea}{\end{array}}

\newcommand{\be}{\begin{equation}}
\newcommand{\ee}{\end{equation}}
\newcommand{\abs}[1]{\lvert#1\rvert}
\newcommand{\ds}{\displaystyle}

\newcommand{\mc}{\mathcal}

\newcommand{\tcr}{\textcolor{red}}

\newcommand{\ov}{\overline}

\newcommand{\B}{\mathcal{B}}
\def\1{\boldsymbol{1}}

\newcommand{\R}{\mathbb{R}}

\newcommand{\se}{\text{ if }}
\newcommand{\tcb}{\textcolor{black}}

\tikzstyle{v_c}=[circle, draw,inner sep=2pt, minimum width=12pt]
\tikzstyle{v_a}=[circle, draw,inner sep=2pt, minimum width=12pt]
\tikzstyle{edge} = [draw,thick,-,font=\small ]
\tikzstyle{label} = [draw,fill=black,font=\normalsize]

\usepackage{booktabs}
\usepackage{tabularx}
\usepackage{array}

\DeclareMathOperator*{\argmax}{argmax}

\def\R{\mathbb{R}}

\def\BibTeX{{\rm B\kern-.05em{\sc i\kern-.025em b}\kern-.08em
    T\kern-.1667em\lower.7ex\hbox{E}\kern-.125emX}}
\usepackage{threeparttable}
    
\title{On {\black Signed Network Games with Binary Actions}}
\author{Martina Vanelli,~\IEEEmembership{Member,~IEEE}, Laura~Arditti, 
Giacomo~Como,~\IEEEmembership{Member,~IEEE,} and
        Fabio~Fagnani
\thanks{Some of the results in this paper appeared in preliminary form in \cite{vanelli2020games,arditti2021equilibria}.}
\thanks{
		Martina Vanelli is with the Institute for Information and Communication Technologies, Electronics and Applied Mathematics (ICTEAM), Université catholique de Louvain, 1348 Ottignies-Louvain-la-Neuve, Belgium
	(e-mail: martina.vanelli@uclouvain.be).}
\thanks{	Laura Arditti was with the Department of Mathematical Sciences
	“G.L. Lagrange,” Politecnico di Torino, 10129 Torino, Italy. She is now
	working in the financial sector in Zurich 8048, Switzerland (e-mail:
	laura.arditti@gmail.com).}
\thanks{	Giacomo Como is with the Department of Mathematical Sciences
	“G.L. Lagrange,” Politecnico di Torino, 10129 Torino, Italy, and also
	with the Department of Automatic Control, Lund University, 22100 Lund,
	Sweden (e-mail: giacomo.como@polito.it).}
\thanks{	Fabio Fagnani is with the Department of Mathematical Sciences
	“G.L. Lagrange,” Politecnico di Torino, 10129 Torino, Italy (e-mail:
	fabio.fagnani@polito.it).}

\thanks{
}
}

\begin{document}

\maketitle
\thispagestyle{empty}

%

\begin{abstract}
We study binary-action pairwise-separable \tcb{graphical} games that encompass both  \tcb{coordination and anti-coordination network games}. 
Our model is grounded in an underlying directed signed graph, where each link is associated with a \tcb{signed} weight that describes both  nature and the strength of the \tcb{strategic pairwise} interaction. \tcb{Specifically, positive link weight corresponds to a strategic complement type interaction, whereas negative link weight corresponds to strategic substitute type interaction.} The utility for each \tcb{player} is then an aggregation of pairwise terms determined by the weights of the signed graph in addition to an individual bias term. 

We consider a scenario that assumes the presence of a prominent cohesive subset of players, who are either connected exclusively by positive weights, or form a structurally balanced subset that can be bipartitioned into two adversarial subcommunities with positive intra-community and negative inter-community edges. 
Under suitable properties of the game restricted to the remaining players, our results guarantee the existence of Nash equilibria characterized by either consensus or polarization within the first group, as well as their stability under  best response transitions. Our results can be interpreted as robustness results, building on the super-modular properties of \tcb{network} coordination games and on a novel use of the concept of graph cohesiveness. 
\end{abstract}


\section{Introduction}
A key feature of many socio-technical systems is the heterogeneity among the behaviors of the \tcb{player}s in the network and among their interactions and mutual influences. Such heterogeneities pose significant challenges in various fields including traffic and routing games and  epidemic models.

In this paper, we focus on heterogeneous interactions within networks of \tcb{player}s engaged in strategic games with binary actions. Our model encompasses two prominent families of network games \cite{Blume:1993,Jackson:2008,Galeotti.ea:2010,Jackson.Zenou:2015}: network coordination games \cite{Ellison:1993,Young:1993,Morris:2000,Young:2006,Jackson.Storms:2025} and network anti-coordination games \cite{Bramoulle.ea:2004,Galam:2004,Bramoulle:2007,Lopez-Pintado:2009,Grabisch.Li:2019},  both of which have a variety of applications in economics, social sciences, and biology. In their simplest version, the utility of a player in a network coordination (anti-coordination) game is an affine increasing (decreasing) function of the number of her neighbors in the network playing the same action. Despite their apparent similarity, both fundamental properties and applications of network coordination and network anti-coordination games are quite different. 	

Network coordination games model the so called strategic complements effects, that is when the choice of a certain action by one player makes it more appealing for other players to play the same action. They are used to model social network features like the adoption of beliefs or behavioral attitudes, or economic ones such as the spread of a new technology. Mathematically, they belong to the broader class of super-modular games \cite{Topkins:1979,Milgrom.ea:1990,Vives:1990,Topkins:1998}. As a consequence, \tcb{pure strategy }Nash equilibria always exist and one special instance of them are the consensus \tcb{action profile}s, namely, those where all individuals are playing the same action. Moreover, asynchronous best response dynamics in network coordination games globally reach the set of \tcb{pure strategy} Nash equilibria \tcb{\cite{Ramazi.Cao:2020, sakhaei2023equilibration}}.   

In contrast, network anti-coordination games are representative of another class of games exhibiting the so called strategic substitutes effect. In this case, the choice of a certain action by one player makes it more appealing for the other players to play the opposite action. They provide a natural model in situations where players are competing for resources that can become congested or in models where players can provide a public good, buy snob goods, or, in general, when there are gains from differentiation. In contrast to network coordination games, existence of \tcb{pure strategy} Nash equilibria for anti-coordination games is guaranteed  in special cases  strongly dependent on the network structure \tcb{\cite{Bramoulle:2007}.} Moreover, even when \tcb{pure strategy Nash} equilibria exist, asynchronous best response dynamics in network anti-coordination games  may get trapped in limit cycles. 

Network games comprising both coordinating and anti-coordinating interactions have been recently proposed to model the presence of anti-conformist behaviors in a social community, accounting for some form of heterogeneity \cite{Ramazi.Riehl.Cao:2016,Grabisch.ea:2019,ramazi2023characterizing, le2023heterogeneous}. More generally, games exhibiting both strategic complements and substitutes have been recently proposed in the economic literature  \cite{Monaco:2016} to model heterogeneous interactions, e.g., markets with coexistence of both Cournot and Bertrand type firms. Such mixed games may fail to admit Nash equilibria. A fundamental example is the matching pennies game, which is a two-player game with one coordinating and one anti-coordinating player. 
In contrast with the matching pennies game, one can imagine that, in a scenario where most of the players are coordinating and form a ``well connected'' subset, the presence of few anti-coordinating players should not prevent the coordinating players to reach a consensus and possibly the whole system to reach a Nash equilibrium. This is one main motivation to our work. 

Another application is when a \tcb{sub-graph of the considered network} is structurally balanced. A structurally balanced \cite{Harary:1953,Cartwright:1956} signed graph is one where vertices can be split into two subsets so that intra links on every set have positive weight, while links connecting the two groups have negative weight. Such property envisages a polarization in the system's equilibrium. 
 In this case we want to determine whether such polarization is still reached even if the entire graph is not structurally balanced. 
 
{\color{black} Network models with signed weights have appeared in many other different fields: to model the presence of antagonistic interactions in social networks \cite{Harary:1953,Cartwright:1956, Macy:2003, Leskovec:2010}, inhibitory signals in genetics \cite{Kauffman:1969, Thomas:1973} and neural networks \cite{Hopfield:2010}, antiferromagnetic bonds in spin glasses \cite{Mezard}.
The popular Linear Threshold Model (LTM) originally introduced in \cite{Granovetter:1978} for non negative graphs has been recently proposed in the context of signed graphs in \cite{He:2013, Golesa:2025}. This dynamical system is strictly related to the best response dynamics for network \tcb{coordination} games. Literature on LTM on signed graphs has focused on showing the differences with respect to the nonnegative case: lack of cascades, dependence on the activation pattern, polarization. Somewhat analogous are the studies of the linear \tcb{and nonlinear} averaging dynamics on signed graphs \cite{Altafini:2012, Altafini:2013, fontan2017multiequilibria, fontan2021role}. Therein, the important concept of structurally balanced graph 
is used to determine the structure of the steady state dynamics. 
}

In this paper, we consider a finite set of \tcb{player}s whose network of interactions is modelled as a directed signed graph. Given two \tcb{player}s, it is possible that a link exists in just one direction and, even when both links  are present, they may have a different weight and possibly weights with opposite signs. This means that \tcb{player} $i$ may tend to coordinate with \tcb{a player} $j$, while \tcb{player} $j$ tends to anti-coordinate with $i$. The utility of \tcb{player} $i$ is an aggregation of a family of pairwise terms one for each of the out-neighbors of $i$ plus an individual bias term. Each pairwise term can be of coordination or anti-coordination type and is modulated by the corresponding graph weight.  We call such games {\black signed network games (SNGs) with binary actions}.
The focus of this paper is on the existence of pure strategy Nash equilibria\footnote{\tcb{We only focus on pure strategy Nash equilibria and do not consider broader equilibrium notions such as mixed strategy or correlated equilibria. This is 
		standard in the literature of coordination and anti-coordination network games, as such equilibria correspond to deterministic action configurations on the network (e.g., technology adoption, opinion formation, behavioral choices).}} in {\black SNGs} and the analysis of their stability with respect to best response 
transitions. 

Our work builds on two fundamental concepts. The first one is a novel use of the notion of cohesiveness originally proposed in the pioneering work \cite{Morris:2000} to describe the \tcb{pure strategy} Nash equilibria of network coordination games. The second is the super-modularity property \cite{Topkins:1979,Milgrom.ea:1990,Vives:1990,Topkins:1998} of network coordination games, particularly the robustness results recently appeared in \cite{arditti2024robust}. 
The general setup of our results is that of a {\black SNG} where the set of \tcb{player}s $\mc V$ is split into two subsets $\mc V=\mc R\cup\mc S$. \tcb{Player}s in $\mc R$ are assumed to be intra themselves coordinating and to form a sufficiently cohesive subset. Alternatively, they are assumed to form a cohesive structurally balanced subset. In this latter case, a transformation of the \tcb{action profile} set allows one to obtain a coordinating subset. If \tcb{player}s in $\mc S$ possess an equilibrium conditioned to 
\tcb{the value} of the actions taken by the \tcb{player}s in $\mc R$, 
then a Nash equilibrium exists that is a consensus or a polarization (respectively) on the \tcb{player}s in $\mc R$. This is the content of our first result,  Theorem \ref{theo:main-existence}\tcb{, with Proposition \ref{prop:mixed1} addressing the special case in which the players in $\mc R$ are coordinating.} \tcb{Corollaries \ref{th:mixed_existence1} and \ref{th:mixed_existence2} identify broad classes of subnetworks for which the condition on $\mc S$ is automatically satisfied,  that is, when the graph restricted to the players in $\mc S$ is structurally balanced or undirected.}

Conditions for the convergence of the best response dynamics are characterized in terms of a novel notion of indecomposability, related to the uniform non-cohesiveness property used in \cite{Morris:2000}, and make a fundamental use of robustness properties of super-modular games introduced in \cite{arditti2024robust}. 
Our main result, Theorem \ref{theo:main-stability}, formalizes these ideas. 	\tcb{
Intuitively, indecomposability rules out the possibility of partitioning the set $\mathcal R$ into distinct polarized subsets that remain in equilibrium under suitable external influences from $\mathcal S$. Although indecomposability is more restrictive than cohesiveness, it yields substantially stronger conclusions: global reachability ensures that from every initial action profile there exists a best-response path leading to a pure strategy Nash equilibrium, while global stability additionally guarantees the existence of sets of pure strategy Nash equilibria that are invariant under best response. These results are particularly significant because SNGs are, in general, neither potential nor supermodular, and therefore do not typically admit global stability or monotonicity arguments.  Despite the stronger requirement imposed on the core subnetwork, the assumptions on the subnetwork restricted to the complementary set $\mathcal S$ remain relatively mild. In particular, the second condition of Theorem \ref{theo:main-stability} is automatically satisfied in broad classes of interaction structures, as shown by Corollaries \ref{coro:reach} and \ref{coro:stability}, notably when the subgraph is structurally balanced for reachability and undirected for stability. A summary of our preliminary and main results is presented in Table \ref{tab:main-results}. }

\tcb{The proposed decomposition $\mc V=\mc R\cup\mc S$ naturally models heterogeneous populations  where a core group of players in $\mc R$ is coordinating (or if the subgraph restricted to $\mc R$ is structural balanced) and the remaining players in $\mc S$ represent perturbing or heterogeneous influences, such as anti-coordinating players, mixed-interaction players, or more generally exogenous disturbances in the network structure. Mixed network coordination/anti-coordination  games 
	arise as a natural special case of this decomposition and motivate the a priori distinction between the subsets $\mc R$ and $\mc S$. In particular, Corollary \ref{th:mixed_existence2} and Corollary \ref{coro:stability} directly apply to these models, establishing the existence of Nash equilibria whenever the coordinating players form a cohesive set and the anti-coordinating interactions are undirected, and the existence of a stable subset whenever $\mc R$ is also indecomposable.} Weaker preliminary results have appeared in \cite{vanelli2020games ,arditti2021equilibria}. In \cite{vanelli2020games}, the underlying graph of interactions was a complete one. In \cite{arditti2021equilibria}, each \tcb{player} was engaged in interactions of only one type, either coordinating or anti-coordinating, and the graph restricted to each of the two subgroups was assumed to be undirected.

We conclude this section by presenting a brief outline of this paper. The remainder of this section is devoted to the introduction of some basic notational conventions to be followed throughout the paper.  In Section \ref{sec:network-coordination-games}, we formally introduce {\black SNG} 
and we present few examples illustrating the problem we want to consider. In Section \ref{sec:prel}, we derive some preliminary results for three classes of signed graphs: undirected, unsigned and structurally balanced. In Section \ref{sec:results}, we present our main results, Theorems \ref{theo:main-existence} and \ref{theo:main-stability}, and show their applicability. 
Finally, Section \ref{conclusion} presents some final remarks.

 \begin{table*}[t]
	\begin{threeparttable}
		\black 	\centering
		\caption{\black Summary of preliminary and main results }
		\label{tab:main-results}
		\renewcommand{\arraystretch}{1.7}
		\begin{tabular}{p{5.5cm} p{4.8cm} p{5cm} p{1.5cm} }
			\hline
			\multicolumn{4}{c}{\textit{Preliminary results
			}}\\
			\hline
			\textbf{Assumption on $\mc G$} &
			\textbf{Property} &
			\textbf{Conclusion}&
			\textbf{Result}\\
			\hline
			Undirected  &
			Potential game & $\mc N$ 
			admits a globally stable subset
			&Proposition~\ref{pr:pot_mixed}	 \\ 		
			Unsigned & Supermodular game& $\mc N$ 
			is globally reachable &
			Proposition~\ref{prop:super-modular} \\
			Structurally balanced & Transformed via $\sigma$ 
			into unsigned $\mc G^{[\sigma]}$ &
			 $\mc N$ 
			is globally reachable & 
			Proposition~\ref{prop:structurally-balanced-equilibria} \\
			\hline
			\multicolumn{4}{c}{  \textit{Main results}\tnote{1}}  \\
			\hline
			\textbf{Assumptions on $\mc G_{\mc R}$} & 
			\textbf{Assumptions on $\mc G_{\mc S}$ \tnote{2} } &
			\textbf{Conclusion} &
			\textbf{Result} \\
			\hline
			Unsigned and cohesive as in \eqref{assumptioni} for some action $a\in \{\pm 1\}$ 
			&
			$\mc N_{\mc S}^{(a\1)} \neq \emptyset$ 
			&
			$\exists$ $x^*\in\mc N$ with $x^*_{\mc R}=a\1$ & 		Proposition~\ref{prop:mixed1}  \\
			Structurally balanced  with transformed graph $\mc G^{[\sigma]}_\mc R$ unsigned and cohesive as in \eqref{wRitau} 
			&
			$\mc N_{\mc S}^{(\tau )}\neq \emptyset $ with $\tau=\sigma_\mc R$ &
			$\exists$ $x^*\in\mc N$ with $x^*_\mc R=\tau$& Theorem \ref{theo:main-existence}
			\\
			Unsigned, indecomposable  (Def. \ref{def:indecomp} and  \eqref{h+-}) and cohesive as in \eqref{wRi} 
			&
			$\mc N_{\mc S}^{(\pm \1)}$ globally reachable (resp. admits a globally stable subset) &
			The set of $x^*$ in $\mc N$  such that $x^*_{\mc R}\in\{\pm \1\}$  is globally reachable (resp. admits a globally stable subset) & Proposition~\ref{prop:stability}
			\\
			Structurally balanced with $\mc G^{[\sigma]}_{\mc R}$ unsigned,  indecomposable and  \eqref{h+-}) and cohesive as in \eqref{wRi} &
			$\mc N_{\mc S}^{(\pm \tau)}$ globally reachable (resp. admits a globally stable subset) &
			The set of $x^*$ in $\mc N$
			such that $x^*_{\mc R}\in \{\pm \tau\} $  is globally reachable (resp. admits a globally stable subset)&
			Theorem~\ref{theo:main-stability}
			\\
			\hline
		\end{tabular}
		\begin{tablenotes}
			\item[1] The player set is $\mc V=\mc R\cup \mc S$. For $\mc U\in\{\mc R, \mc S\}$,  $\mc G_\mc U$ denotes the graph restricted to $\mc U$  and $\mc N^{y}_\mc S$ the set of NE of the game on $\mc S$ with actions of $\mc R$ fixed to $y$.
			\item[2] Notably, these assumptions are satified when $\mc G_{\mc S}$ is structurally balanced or undirected (see Corollaries \ref{th:mixed_existence1}, \ref{th:mixed_existence2},  \ref{coro:reach}, and \ref{coro:stability}).\end{tablenotes}
	\end{threeparttable}
\end{table*}

\textbf{Notational conventions}
Let $\mathds{R}$ and $\R_{+}$ denote the set of real numbers and non-negative real numbers, respectively. For a \tcb{nonempty} finite set \tcb{$\mc{V}$}, $\R^{\mc{V}}$  denotes the space of real column vectors $x$, whose entries $x_i$ are indexed by the elements $i$ in $\mc{V}$ \tcb{(e.g., $x=(x_1, x_2, x_3)$ for $\mc V=\{1,2,3\}$)}. For a vector $x$ in $\R^{\mc V}$ and a subset $\mc A\subseteq\mc V$, we use the notation $x_{\mc A}$ in $\R^{\mc A}$ for the restricted vector with entries $(x_{\mc A})_i=x_i$ for every $i$ in $\mc A$ \tcb{(e.g., $x_\mc A=(x_1, x_3)$ for $\mc A=\{1,3\}$)}. The all-$1$ vector, whose size may be deduced from the context, will be denoted by $\1$. For a vector $x$ in $\R^{\mc V}$ and some $i$ in $\mc V$, we write $x_{-i}=x_{\mc V\setminus\{i\}}$ for the vector in $\R^{\mc V\setminus\{i\}}$ obtained from \tcb{$x$} by removing its $i$-th entry \tcb{(e.g., $x_{-2}=x_{\{1,3\}}=(x_1, x_3)$)}. 
Similarly, for two finite sets $\mc U$ and $\mc V$, $\R^{\mc{U}\times\mc V}$  denotes the space of real matrices $M$ whose entries $M_{ij}$ are indexed by the pairs $(i,j)$ in $\mc{U}\times\mc V$. For a matrix $M$ in $\R^{\mc{U}\times\mc V}$ and two subsets $\mc A\subseteq\mc U$ and $\mc B\subseteq\mc V$, we use the notation $M_{\mc A\mc B}$ in $\R^{\mc A\times\mc B}$ for the restricted matrix, with entries $(M_{\mc A\mc B})_{ij}$ for every $i$ in $\mc A$ and $j$ in $\mc B$.  
For a vector $d$ in $\R^{\mc A}$, $[d]$stands for the diagonal matrix with diagonal coinciding with $d$, i.e.,  $[d]\in\R^{\mc A\times\mc A}$ is such that $[d]_{ii}=d_i$ and $[d]_{ij}=0$ for every $i$ in $\mc A$ and $j$ in $\mc A\setminus\{j\}$.

\section{Model}
\label{sec:network-coordination-games}

In this section, we introduce the signed network game model in its general form and present the main issues addressed in the rest of the paper.
\subsection{{\black Signed network games with binary actions}}
We model networks as finite directed weighted signed graphs $\mc G=(\mc V, \mc E, W)$, with non-empty set of nodes $\mc V$, set of directed links $\mc E\subseteq\mc V\times\mc V$, and weight matrix $W$ in $\R^{\mc V\times\mc V}$ such that $W_{ij} \neq 0$ if and only if $(i,j)$ is a link in $\mc E$. 
We do not allow for the presence of self-loops, equivalently, we assume that the weight matrix $W$ has zero diagonal. Throughout, finite directed weighted signed graphs without self-loops will be simply referred to as \emph{networks}.\footnote{Notice that we do not assume in general that $W_{ij}$ and $W_{ji}$ have the same sign (a property referred to as ``digon sign-symmetry'' in \cite{Altafini:2013}.)}

We shall call a network $\mc G=(\mc V, \mc E, W)$ \emph{undirected} when the weight matrix is symmetric, i.e., $W=W'$,  so that in particular $(j,i)$ is a link in $\mc E$ whenever $(i,j)$ is, and in this case they have the same weight $W_{ij}=W_{ji}$. We refer to a network $\mc G=(\mc V, \mc E, W)$ as \emph{unsigned} when the weight matrix $W$ is non-negative, i.e., all links $(i,j)$ in $\mc E$ have positive weight $W_{ij}>0$. For a network $\mc G=(\mc V, \mc E, W)$ and a subset of nodes $\mc U\subseteq\mc V$, we consider the subnetwork $\mc G_{\mc U}=(\mc U,\mc E_{\mc U},W_{\mc U\mc U})$ where $\mc E_{\mc U}=\mc E\cap(\mc U\times\mc U)$ is the subset of links with both tail node and head node in $\mc U$, while $W_{\mc U\mc U}$ is the sub-matrix of $W$ obtained by restricting its row and column sets to $\mc U$.

\begin{example}\label{ex1}
Figure \ref{fig:graph1} illustrates a network with weight values reported next to the corresponding links. Notice that, for $\mc R=\{1,\dots, 8\}$, 
the \tcb{subnetwork} $\mc G_{\mc R}$ 
\tcb{is} unsigned, whereas for  $\mc S=\mc V\setminus \mc R=\{9, \dots, 13\}$ 
the \tcb{subnetwork} $\mc G_{\mc S}$ 
\tcb{is} undirected.
\end{example}
	\begin{figure}
	\centering
	\begin{tikzpicture}[scale=1.2]
		\foreach \x/\name in {(1.3,0)/1,  (2.7,0)/3, (4,-1)/4, (4,0)/5,(5,-1)/6, (2,1)/2}\node[shape=circle,draw, fill=gray!30!white](\name) at \x {\small\name};
		
			\foreach \x/\name in {(2,-1)/7, (3.5,1)/8/-1}\node[shape=circle,draw, fill=gray!30!white ](\name) at \x {\small\name};
		\foreach \x/\name in {(6,0)/9,(3,-1.5)/10,(1,-1.5)/11,(0,0)/12,(5,1)/13}\node[shape=circle,draw ](\name) at \x {\small\name};
		
		\foreach \a/\b/\w in {4/6/2,1/7/2, 2/8/3,7/3/1}\path [<-,draw] (\a) edge node[above] {\small\w} (\b);
		
		\foreach \a/\b/\w in {  7/10/-2,3/13/-7,  13/5/-2 }\path [<-,draw, red] (\a) edge node[above] {\small\w} (\b);

		\foreach \a/\b/\w in {3/5/3,  4/10/2,3/4/1}\path [->,draw] (\a) edge node[below] {\small\w} (\b);
		
		\foreach \a/\b/\w in {  1/12/2}\path [->,draw, bend right=20] (\a) edge node[above] {\small\w} (\b);
		
		\foreach \a/\b/\w in {5/9/-1}\path[->,draw,red, bend left=20] (\a) edge node[above] {\small\w} (\b);

		\path [->,draw,bend left=20] (9) edge node[below] {\small4} (5);
		
		\path [->,draw,red, bend right=20] (12) edge node[below] {\small-4} (1);
		
		
		\foreach \a/\b/\w in {2/3/1, 4/5/3}\path [<->,draw] (\a) edge node[right] {\small\w} (\b);

		\foreach \a/\b/\w in { 12/11/-2,13/9/-5}\path [<->,draw,red] (\a) edge node[right] {\small\w} (\b);
		
		\foreach \a/\b/\w in { 11/10/-1}\path [<->,draw,red] (\a) edge node[below] {\small\w} (\b);
		
		\foreach \a/\b/\w in { 8/13/-1,7/11/-1}\path [<->,draw,red] (\a) edge node[above] {\small\w} (\b);
		
		\foreach \a/\b/\w in {1/2/3}\path [->,draw] (\a) edge node[left] {\small\w} (\b);		
	\end{tikzpicture}
	\caption{A network with 13 nodes. Weights are represented by the values on the links. Nodes in the set $\mc R= \{1, \dots, 8\}$ are colored in gray, while nodes in the complement set $\mc S= \{9, \dots, 13\}$ are white.}
	\label{fig:graph1}
\end{figure}
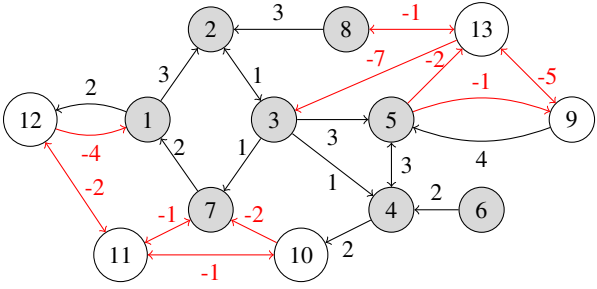
\tcb{We model the strategic interaction as a binary-action game on a signed network, where each player’s utility is the sum of two parts: an intrinsic bias toward one of the two actions and a network interaction term determined by the signed influence of neighboring players. The resulting utility function extends standard network coordination and anti-coordination games by allowing both positive and negative interactions as well as heterogeneous thresholds. 
Given a network $\mc G=(\mc V,\mc E,W)$, we associate to each node $i$ in $\mc V$ a scalar value $h_i$ in $\mathbb{R}$, representing the intrinsic bias of player $i$, and stack all such values in a vector $h$ in $\mathbb{R}^{\mc V}$, to be referred as the \emph{external field}.} 
 \begin{definition}\label{def:coordination} The {\black \emph{signed network game (SNG) }} with \emph{binary actions}  on a network $\mc G=(\mc V,\mc E,W)$ with  external field $h$ in $\R^{\mc V}$ is the strategic game with player set $\mc V$, whereby every player $i$ in $\mc V$ has action set $\mc A=\{\pm1\}$ and utility function $u_i: \mc X\rightarrow \R$ given by
	\be\label{eq:ut} u_i(x)=u_i(x_i,x_{-i})=h_ix_i+ x_i\sum_{j\in\mc V}W_{ij}x_j\,,\ee 
	for every \tcb{action profile} $x$, where $\mc X=\mc A^\mc V$ denotes the set of \tcb{action profile}s of all players and $x_{-i}$ in $\mc X_{-i}=\mc A^{\mc V\setminus\{i\}}$ stands for the \tcb{action profile} of all players except for player $i$.
\end{definition}
We now provide an interpretation of the {\black SNG}s introduced above. The utility function \eqref{eq:ut} is the sum of two terms. 
The first addend in the right-hand side of  \eqref{eq:ut} is a standalone term $h_ix_i $ (i.e., a term that depends only on the action of player $i$) that models the bias of player $i$ towards one of the two possible actions depending on the sign of the preference term $h_i$: if $h_i>0$ (respectively, $h_i<0$), then player $i$ has a bias towards action $+1$ ($-1$), in that, in the absence of the network (i.e., if $W=0$) this action would provide a utility that is $2h_i$ larger that its alternative $-1$ ($+1$); if $h_i=0$, then no bias is present in that both actions $+1$ and $-1$ would reward player $i$ with the same utility in the absence of the network.
The second addend in the right-hand side of  \eqref{eq:ut} is the aggregate of pairwise interaction terms $x_iW_{ij}x_j$, each of which can be interpreted as the utility of player $i$ in a two-player game with player $j$. Notice that, while the sum index $j$ in the right-hand side of \eqref{eq:ut} formally runs over all players $j$ in $\mc V$, the sum is effectively only over players $j$ such that $W_{ij}\ne0$, i.e., out-neighbors of player $i$ in the network $\mc G$. Observe that a {\black SNG} with binary actions is an instance of a polymatrix or pairwise-separable game \cite{Yanovskaya:1968,Arditti.Como.Fagnani:2024}.

\tcb{
\begin{example}
		Consider the network (a) in Figure \ref{fig:disc_game} with external field $h$.  The corresponding SNG has action set $\mathcal{A}=\{\pm1\}$ and utility functions
$$u_1(x_1,x_2)=h_1 x_1 + x_1 x_2\,,\qquad u_2(x_1,x_2)=h_2 x_2 - x_1 x_2\,.$$	Player $1$ (respectively, player $2$) experiences a positive (negative) interaction with player $2$ ($1$), reflecting coordination and anti-coordination incentives, respectively. The external field $h_i$ induces a bias: if $|h_i|\ge 1$, player $i$ is effectively a stubborn agent always preferring action $\mathrm{sign}(h_i)$ regardless of the neighbor’s action. When $h=0$, the game reduces to a standard discoordination game \cite{Jackson.Zenou:2015}.
	\end{example}} 

\begin{figure}
	\centering
	\begin{tikzpicture}[scale=0.9, swap]
		\node[v_c] (vc) at (0,0) {1};
		\node[v_a, draw] (va) at (2,0) {2};
		\path[->, draw, bend right=20] (vc) edge node[below] {1} (va);
		\path[<-, draw, bend left=20,red] (vc) edge node[above] {-1} (va);
		\node (1) at (1,-1) {(a)};
	\end{tikzpicture}
	\hspace{15pt}
	\begin{tikzpicture}[scale=0.9, swap]
		\node[shape=circle,draw] (1) at (0,0) {1};
		\node[shape=circle,draw] (2) at (1,1.5) {2};
		\node[shape=circle,draw] (3) at (2,0) {3};
		\path[->, draw,red] (1) edge node[left] {-1} (2);
		\path[->, draw,red] (2) edge node[right] {-1} (3);
		\path[->, draw,red] (3) edge node[below] {-1} (1);
		\node (1) at (1,-1) {(b)};
	\end{tikzpicture}	
	\caption{Two signed graphs representing: (a) the discoordination game and (b) a directed anti-coordination game.}
	\label{fig:disc_game}
\end{figure}
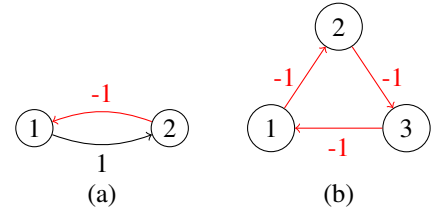

We say that a player $i$ in $\mc V$ has a \textit{coordinating} interaction with another player $j$ when $W_{ij}>0$, meaning that the utility $u_i(x)$ of player $i$ increases when player $j$ plays the same action as $i$, i.e., when $x_ix_j=1$. Conversely, when $W_{ij}<0$, player $i$ is said to have an \textit{anti-coordinating} interaction with player $j$ in $\mc V$, namely its utility decreases when player $j$ is playing the same action. 
In the special case when the network $\mc G=(\mc V,\mc E,W)$ is unsigned, namely when $W_{ij}\geq 0$ for every $i$ and $j$ in $\mc V$, the {\black SNG} with binary actions is known as a \textit{network coordination game}. On the other hand, when $W_{ij}\leq 0$ for every $i$ and $j$ in $\mc V$, the {\black SNG} is known as a \textit{network anti-coordination game}. \tcb{Another special case is when the rows of the network's weight matrix $W$ have constant sign, i.e., when $W_{ij}W_{ik}\ge0$ for every $i$, $j$, and $k$ in $\mc V$: this is the case when every player $i$ in $\mc V$ has either all coordinating interactions or all anti-coordinating interactions with her neighbors that was studied in our previous works \cite{vanelli2020games} and \cite{arditti2021equilibria}.}

Observe that our general model of {\black SNG} accounts for asymmetric behaviors: while it is possible that two nodes $i$ and $j$ have the same type of interaction with one another, it can happen that a player $i$ has a coordinating interaction with another player $j$ (i.e., $W_{ij}>0$) while $j$ has an anti-coordinating interaction with $i$ or is not influenced by player $i$ at all (i.e., $W_{ji}=0$). This makes the analysis of such games particularly challenging.

\subsection{Best Responses and Nash Equilibria}

The \textit{best response} (BR) correspondence for a player $i$ in $\mc V$ is the set of its optimal actions given the \tcb{action profile} or the other players. 
Formally,\footnote{\tcb{Observe that, since the utility function of player $i$ depends only on the actions $x_i$ of player $i$ and on the action profile $x_{\mc N_i}$ of the players her neighborhood $\mc N_i=\{j\in\mc V:W_{ij}>0\}$ in the network, so does the best response $\mc B_i(x_{-i})$. Specifically, one has that, if $x$ and $y$ in $\mc X$ are two action profiles such that $x_{\mc N_i}=y_{\mc N_i}$, then $\mc B_i(x_{-i})=\mc B_i(y_{-i})$. Such locality property of the best response correspondence is a common feature of all graphical games (c.f.~\cite[Ch.9]{Jackson:2008} or \cite{Galeotti.ea:2010,Jackson.Zenou:2015}).}}
$$\B_i(x_{-i})= \argmax_{x_i \in\mc A}u_i(x_i,x_{-i}) \,.$$
A (pure strategy) \emph{Nash equilibrium} is a \tcb{action profile} $x^*$ in $\mc X$ such that 
$$x^*_i\in\B_i(x^*_{-i})\,,\qquad\forall i\in\mc V\,.$$
The set of Nash equilibria of a game will be denoted as $\mc N$. 
A Nash equilibrium $x^*$ in $\mc N$ is called \emph{strict} if $\mc B_i(x^*_{-i}) =\{x^*_i\}$ for every $i$ in $\mc V$. The set of strict Nash equilibria of a game will be denoted as $\mc N^*$. 

The purpose of this work is to investigate the existence of Nash equilibria for {\black SNG}s and study their stability properties under best response transformations. Formally, consider two \tcb{action profile}s $x$ and $y$ in $\mc X$. For $l\ge0$, a length-$l$ \emph{best response path} (\emph{BR-path}) from $x$ to $y$ is an $(l+1)$-tuple of \tcb{action profile}s in $\mc X^{l+1}$, denoted with 
$(x^{(0)},x^{(1)},\ldots x^{(l)})$,  such that:
\begin{itemize}
	\item $x^{(0)}=x$, and $x^{(l)}=y$; 
	\item for every $k=1,2,\ldots, l$, there exists a player $i_k$ in $\mc V$ such that  
	\be\label{eq:BRpath}\!\!\!\!\!\!\! x^{(k)}_{-i_k}=x^{(k-1)}_{-i_k}\,, \qquad 	x_{i_k}^{(k)}\in\B_{i_k}(x_{-i_k}^{(k-1)})\setminus\{x^{(k-1)}_{i_k}\}\,.\ee 
\end{itemize}
A BR-path is thus a sequence of unilateral modifications of the \tcb{action profile} where one \tcb{player} at a time changes its action to any of its best responses to the current \tcb{action profile} of the rest of the population.\footnote{\tcb{Observe that the indices of the updating players $(i_1, i_2, \ldots, i_l)$ in a BR-path trace out an underlying activation sequence in the sense of \cite{sakhaei2023equilibration}.}} 
Notice that, with the above definition, for every \tcb{action profile} $x$ in $\mc X$,   there always exists a length-$0$ BR-path $(x)$ from $x$ to itself. 

We shall refer to a subset of \tcb{action profile}s $\mc X^*\subseteq\mc X$ as:
\begin{itemize}
	\item \emph{BR-reachable} from strategy $x$ in $\mc X$ if there exists a BR-path from $x$ to some \tcb{action profile} $y$ in $\mc X^*$;  
	\item \emph{globally} BR-reachable if it is BR-reachable from every \tcb{action profile} $x$ in $\mc X$;
	\item \emph{BR-invariant} if there are no BR-paths from any $y$ in $\mc X^*$ to any $z$ in $ \mc X\setminus\mc X^*$;
	\item \emph{globally BR-stable} if it is globally BR-reachable and BR-invariant.
\end{itemize}

Notice that a BR-reachable (hence, in particular a globally BR-reachable or globally BR-stable) set of \tcb{action profile}s $\mc X^*$ is necessarily non-empty. 
Observe that a singleton $\{x^*\}$ is BR-invariant if and only if $x^*$ is a strict Nash equilibrium. 
A globally BR-stable set is one from which it is impossible to exit through unilateral best response modifications and that can be reached from every initial \tcb{action profile} through a finite number  of such modifications.

\begin{remark} The notions of BR-reachability, BR-invariance, and BR-stability are particularly relevant for the study of so-called asynchronous best response dynamics \cite{Blume:1995}.  
These are a class of asynchronous dynamics modeled as discrete-time Markov chains with finite state space coinciding with the \tcb{action profile} set $\mc X$, whereby, at every time step $t=0,1,2,\ldots$, conditioned on the current configuration $X(t)=x$, one player $i$ is randomly selected from  the player set $\mc V$ according to a distribution that assigns positive probability to all players, and she updates her action to a value $X_i(t + 1)$ chosen uniformly from her best response set $\mc B_i(x_{-i})$. In fact, observe that, on the one hand, a subset of \tcb{action profile}s  $\mc X^*\subseteq\mc X$ is  BR-invariant if and only if it is a trapping set for the asynchronous best response dynamics, i.e., if $X(s)\in\mc X^*$ for some $s\ge0$ implies that $X(t)\in\mc X^*$ for every $t\ge s$.  On the other hand, $\mc X^*$ is globally BR-stable if and only if, for every probability distribution of the initial state $X(0)$ of the asynchronous best response dynamics, there exists a random time $T$ that is finite with probability $1$ and is such that $X(t)\in\mc X^*$ for every $t\ge T$. \end{remark}

\begin{remark}\label{remark:I-reachable}
 In our previous work \cite{arditti2024robust}, the notion of improvement path (I-path) was introduced along with the related notions of I-reachability, I-invariance, and I-stability. Specifically, a length-$l$  I-path is a $(l+1)$-tuple of \tcb{action profile}s  $(x^{(0)},x^{(1)},\ldots x^{(l)})$ such that for every $k=1,2,\ldots, l$ there exists a player $i_k$ in $\mc V$ such that  
\be\label{eq:Ipath} x^{(k)}_{-i_k}=x^{(k-1)}_{-i_k}\,, \qquad 	u_{i_k}(x^{(k)})>u_{i_k}(x^{(k-1)}) \,.\ee 
Notice that, if the action space of every player is binary, as is our case, then condition \eqref{eq:Ipath} implies condition \eqref{eq:BRpath}, so that every I-path is also a BR-path, but not necessarily vice versa. As a consequence, every (globally) BR-reachable set of \tcb{action profile}s is (globally) BR-reachable. Conversely, every BR-invariant  set of \tcb{action profile}s is I-invariant, but not vice versa: e.g., every non-strict Nash equilibrium is I-invariant but not BR-invariant. On the other hand, neither I-stability implies BR-stability nor vice versa: e.g., for the {\black SNG} of Example \ref{ex:onelink} with $\alpha=0$, \tcb{the set 
$\mc N$} is globally I-stable but not BR-stable. 
\end{remark}

The following two examples show how Nash equilibria may not exist for {\black SNG}s. {\black In particular, Example \ref{ex:3ringanticoord} illustrates how nonexistence may occur even in \emph{purely anti-coordination} games when the graph is \emph{directed}.}

	\addtocounter{example}{-1}
\begin{example}[cont'd]\label{ex:discoord}
\tcb{The discoordination game, namely, the} {\black SNG} on the network (a) in Figure \ref{fig:disc_game} with external field $h=0$ 
\tcb{and} utilities $u_1(x_1, x_2) = x_1x_2$ 
and $u_2(x_1, x_2) =-x_1x_2$, 
\tcb{admits no pure strategy Nash equilibria}.
\end{example} 
\begin{example}\label{ex:3ringanticoord}
The {\black SNG} with binary actions on the network (b) in Figure \ref{fig:disc_game} with external field $h=0$ is an anti-coordination game. For $x^*$ to be a Nash equilibrium, one should have $x_1^*=-x_2^*=x_3^*=-x_1^*$, which is impossible in $\{\pm1\}^3$. Hence, this game does not admit any pure strategy  Nash equilibria.
\end{example}

Even when \tcb{Nash equilibria of {\black SNG}s exist, they} may fail to be  globally BR-reachable or BR-invariant, as illustrated in the following two examples. 

\begin{example}\label{ex:onelink} Consider a {\black SNG} on a network $\mc G=(\mc V,\mc E,W)$ with two nodes  $\mc V=\{1,2\}$, a single link $\mc E=\{(1,2)\}$, and weight matrix $W=(1,0)(0,1)'\,,$ 
and  external field $h=(0,\alpha)$, where $\alpha$ is a scalar parameter. Then, the set of pure strategy  Nash equilibria is 
$$\mc N=\left\{\ba{lcl}
\{-\1\}&\se&\alpha<0\,,\\
\{\pm\1\}&\se&\alpha=0\,,\\
\{\1\}&\se& \alpha>0\,.
\ea\right.$$
It easily verifiable that $\mc N$ is globally BR-reachable for every $\alpha$. 
In fact, for $\alpha\ne0$ then the unique Nash equilibrium is strict, so that $\mc N$ is BR-invariant, hence BR-stable. On the other hand, for $\alpha=0$, there are two Nash equilibria neither of which is strict: in this case, neither $\mc N$ nor any of it subsets are 
\tcb{BR-invariant}, and there are no BR-stable sets. 
\end{example}

\begin{example} \label{ex:nostability} 

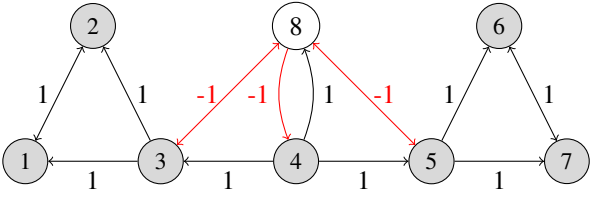
\begin{figure}
	\centering	
	
	\vspace{8pt}
	\begin{tikzpicture}[scale = 0.9]
	
		\node[shape=circle,draw] (8) at (4,2){8} ;
		
		\foreach \x/\name in {(0,0)/1,   (1,2)/2, (2,0)/3,  (4,0)/4,  (6,0)/5,(7,2)/6,  (8,0)/7}\node[shape=circle,draw, fill=gray!30!white](\name) at \x {\small\name};
		
		\path [<->,draw] (1) edge node[left] {1} (2);
		\path [->,draw] (3) edge node[right] {1} (2);
		\path [->,draw] (3) edge node[below] {1} (1);
		\path [<-,draw] (3) edge node[below] {1} (4);
		\path [<->,draw=red, text=red] (3)  edge node[left] {-1} (8);
		\path [<->,draw=red, text=red] (5)  edge node[right] {-1} (8);
		\path[->, draw, bend right=20] (4) edge node[right] {1} (8);
		\path[<-, draw, bend left=20,red] (4) edge node[left] {-1} (8);
		\path [->,draw] (4) edge node[below] {1} (5);

		\path [->,draw] (5) edge node[left] {1} (6);
		\path [->,draw] (5) edge node[below] {1} (7);
		\path [<->,draw] (6) edge node[right] {1} (7);
	\end{tikzpicture}
	\caption{A signed graph with $8$ nodes and a coordinating set made of $7$ players (in gray). 
	}
	\label{fig:noeq_nostability}
\end{figure}	

Consider the signed-network related to the graph in Figure \ref{fig:noeq_nostability} with $h=0$. Observe that the set $\mc R= \{1,\dots, 7\}$ is a coordinating set. A direct check shows that the set of Nash equilibria is $\mc N=\{\pm x^*\}$, where $x^*$ is the \tcb{action profile} with $x^*_{\mc R} =\1$  and $x^*_8 = -1$. However, $\mc N$ is not globally BR-reachable, as the subset of \tcb{action profile}s $$\mc X_{\mc R}=\{y\in\mc X:\,y_1=y_2=y_3=-y_5=-y_6=-y_7\}\,,$$ is BR-invariant. 
	
\end{example}

\section{Preliminary results}\label{sec:prel}

In this subsection, we shall first present some preliminary results on two special classes of {\black SNG}s: {\black SNG}s on undirected networks and network coordination games (i.e., {\black SNG}s on unsigned networks). We will rely on standard results on super-modular and potential games, respectively, to show that the set of Nash equilibria $\mc N$ of both these classes of games is non-empty and globally BR-reachable, and, for undirected {\black SNG}s, that $\mc N$ contains a non-empty globally BR-stable subset. We will then present a number of examples of {\black SNG}s outside these two classes that illustrate the specific challenge in the study of the existence and stability of Nash equilibria of general {\black SNG}s.

\subsection{{\black SNG}s on undirected networks}
A strategic game with \tcb{action profile} space $\mc X$  is called \emph{exact potential} if there exists a \emph{potential function} $\Phi:\mc X\to\R$ 
such that,  for every player $i$ in $\mc V$ and \tcb{action profile}s $x$ and $y$ in $\mc X$,  
\begin{equation}\label{potential_function}
x_{-i}=y_{-i}\quad\Longrightarrow\quad u_i(y)-u_i(x) =\Phi(y)-\Phi(x)\,,
 \end{equation}
It is well known \cite{Monderer.Shapley:1996} that the set of Nash equilibria $\mc N$ of every finite exact potential game  is always nonempty: in particular, every global maximum point $x^*$  of the potential function $\Phi(x)$ is a Nash equilibrium.  The following result states that a {\black SNG} on a network $\mc G$ is an exact potential game if and only if $\mc G$ is undirected and that, in this case, its set of Nash equilibria contains  a nonempty globally BR-stable subset.

 \begin{proposition}\label{pr:pot_mixed}
Consider a {\black SNG} with binary actions on a network $\mc G=(\mc V,\mc E,W)$ with  external field $h$.
Then, 
 \begin{enumerate}
 \item[(i)] the game is exact potential if and only if $\mc G$ is undirected.  \end{enumerate}
Moreover, if $\mc G$ is undirected, then:  
 \begin{enumerate}
 \item[(ii)]  $\Phi:\mc X\to\R$ is an exact potential function for the {\black SNG} if and only if there exists a constant $C$ in $\R$ such that 
\be\label{potential_function_het}
\Phi(x)=\frac{1}{2}\sum_{i,j\in \mc V}W_{ij}x_i x_j + \sum_{i\in \mc V} h_i x_i+C\,,  
\ee
for every \tcb{action profile} $x$ in $\mc X$;
 \item[(iii)]   there exists a globally BR-stable set $\ov{\mc N}$ such that 
 $$\argmax_{x\in\mc X}\Phi(x)\subseteq\ov{\mc N}\subseteq\mc N\,.$$
 \end{enumerate}
\end{proposition} 
\begin{proof}
(i) Consider two \tcb{action profile}s $x$ and $y$ in $\mc X$ such that $x_{-i}=y_{-i}$ and $y_i=1=-x_i$ for some player $i$ in $\mc V$. Then, from \eqref{eq:ut} we have  that 
\be\label{uiy-x} u_i(y)-u_i(x) =2\sum_{k\in\mc V}W_{ik}x_k+2h_i\,.\ee

If $\mc G$ is undirected, then \eqref{potential_function_het} and \eqref{uiy-x} imply that 
\be\label{phiy-x} 
\ba{rcl}\Phi(y)-\Phi(x)&=&\ds\sum_{k\in\mc V}W_{ik}x_k+\sum_{k\in\mc V}W_{ki}x_k+2h_i\\[15pt]
&=&\ds2\sum_{k\in\mc V}W_{ik}x_k+2h_i\\[15pt]
&=&u_i(y)-u_i(x)\,,\ea\ee
thus proving that the {\black SNG} with binary actions on $\mc G$ with  external field $h$ is an exact potential game. 

On the other hand, if $\mc G$ is not undirected, let $i\ne j$ in $\mc V$ be such that \be\label{Wij<Wji}W_{ij}< W_{ji}\,,\ee 
and consider a configuration $x$ such that $x_i=x_j=-1$.  Let $y$, $w$, and $z$ in $\mc X$ be the configuration such that, respectively: $y_i=1$ and $y_{-i}=x_{-i}$;  $w_j=1$ and $w_{-j}=x_{-j}$; $z_{-j}=y_{-j}$ and $z_j=1$. Observe that 
\be\label{uiz-w} \ba{rcl}u_i(z)-u_i(w) &\!\!\!=\!\!\!&\ds2\sum_{k\in\mc V}W_{ik}w_k+2h_i\\[15pt]
&\!\!\!=\!\!\!&\ds2\sum_{k\in\mc V}W_{ik}x_k+2h_i+2\sum_{k\in\mc V}W_{ik}(w_k-x_k)\\[15pt]
&\!\!\!=\!\!\!&u_i(y)-u_i(x)+4W_{ij}\,,\ea\ee
where the last identity follows from the fact that $w_{-j}=x_{-j}$ and $w_i=1=-x_j$. 
Similarly, we have that 
\be\label{ujz-y} \ba{rcl}u_j(z)-u_j(y) &\!\!\!=\!\!\!&2\ds\sum_{k\in\mc V}W_{jk}y_k+2h_j\\[15pt]
&\!\!\!=\!\!\!&2\ds\sum_{k\in\mc V}W_{jk}x_k+2h_j+2\sum_{k\in\mc V}W_{jk}(y_k-x_k)\\[15pt]
&\!\!\!=\!\!\!&u_j(w)-u_j(x)+4W_{ji}\,,\ea\ee
where the last identity follows from the fact that $y_{-i}=x_{-i}$ and $y_i=1=-x_i$. 
It the follows from \eqref{Wij<Wji}, \eqref{uiz-w}, and \eqref{ujz-y} that
\be\label{no-potential}
\ba{rcl}0
&<&4(W_{ji}-W_{ij})\\[5pt]
&=&u_i(y)-u_i(x)+u_j(x)-u_j(w)\\[5pt]
&&+u_i(w)-u_i(z)+u_j(z)-u_j(y)\,.
\ea
\ee
It then follows from \cite[Corollary 2.9]{Monderer.Shapley:1996}  that the {\black SNG} with binary actions on $\mc G$ with external field $h$ is not an exact potential game.

(ii) If $\mc G$ is undirected, then \eqref{phiy-x} implies that $\Phi$ defined in \eqref{potential_function_het} is an exact potential function for  the {\black SNG} with binary actions on $\mc G$ with  external field $h$. The claim then follows from  \cite[Lemma 2.7]{Monderer.Shapley:1996}.

(iii) The claim is a direct consequence of \cite[Lemma 2]{Catalano.ea:2024}. 
\end{proof}


\subsection{{\black SNG}s on unsigned networks}
A finite strategic game where all players have binary action set $\mc A=\{\pm1\}$ is \emph{super-modular}\cite{Topkins:1979,Milgrom.ea:1990,Vives:1990,Topkins:1998} if it satisfies the \emph{increasing difference property}, i.e., if 
\be\label{IDP}u_i(+1,x_{-i})-u_i(-1,x_{-i})\le u_i(+1,y_{-i})-u_i(+1,y_{-i})\,,\ee
for every two \tcb{action profile}s $x$ and $y$ in $\mc X$ such that $x\le y$. 
For super-modular games, Nash equilibria are guaranteed to exist and they form a lattice in the \tcb{action profile} space.
The following result states that a {\black SNG} is super-modular if and only if it is a network coordination game and that, in this case, the set of its Nash equilibria is non-empty and globally BR-reachable. 

\begin{proposition}  \label{prop:super-modular}
Consider a {\black SNG} with binary actions on a network $\mc G=(\mc V,\mc E,W)$ with  external field $h$.
Then, 
 \begin{enumerate}
 \item[(i)] the game is super-modular if and only if $\mc G$ is unsigned;  
 \item[(ii)] if $\mc G$ is unsigned, then the set of Nash equilibria $\mc N$ is globally BR-reachable.
 \end{enumerate}
\end{proposition}
\begin{IEEEproof} (i) For every \tcb{action profile} $x$ in $\mc X$ and player $i$ in $\mc V$, let 
$$\delta_i(x)=u_i(+1,x_{-i})-u_i(\tcb{-}1,x_{-i})=2h_i+2\sum_{j\in\mc V}W_{ij}x_j\,.$$
Then, for every two \tcb{action profile}s  $x$ and $y$ in $\mc X$, we have
\be\label{dy-dx}\delta_i(y)-\delta_i(x)=\sum_{j\in\mc V}W_{ij}(y_j-x_j)\,,\ee
for every $i$ in $\mc V$. 
If $W_{ij}\ge0$ for every $i$ and $j$ in $\mc V$, then \eqref{dy-dx} implies that 
$$\delta_i(y)-\delta_i(x)=\sum_{j\in\mc V}W_{ij}(y_j-x_j)\ge0\,,$$ 
whenever $x\le y$, so that the increasing difference property \eqref{IDP} holds true, hence {\black SNG} is super-modular.  

Conversely, if $W_{ij}<0$ for some $i$ and $j$ in $\mc V$, then let $x$ and $y$  in $\mc X$ be two \tcb{action profile}s such that $x_j=-1$, $y_j=+1$, and $x_{-j}=y_{-j}$, so that $x\le y$. Then,  \eqref{dy-dx} implies that
$$\delta_i(y)-\delta_i(x)=2W_{ij}<0\,,$$ so that the increasing difference property \eqref{IDP} does not hold true, hence the {\black SNG} is not super-modular.  

(ii) This is a direct consequence of \cite[Proposition 3(v)]{arditti2024robust}, which asserts that the set of Nash equilibria is globally I-stable (i.e., globally I-reachable and I-invariant) and the fact that I-reachability implies BR-reachability (c.f.~Remark \ref{remark:I-reachable}) \footnote{\tcb{This is a standard result and could also be derived, for instance, 
		from \cite[Theorem 1]{sakhaei2023equilibration}.}} 
\end{IEEEproof}\medskip

\begin{remark} Proposition \ref{prop:super-modular} ensures that every network coordination game, i.e., every {\black SNG} with binary actions on an unsigned network $\mc G$, admits Nash equilibria and that the set $\mc N$ of Nash equilibria is globally reachable.  Notice that, $\mc N$ may not be invariant, as Example \ref{ex:onelink} illustrates. 
\end{remark}
\subsection{Signed network games on structurally balanced networks}
The results in the previous subsection can be extended by introducing the notion of structural balance \cite{Harary:1953,Cartwright:1956}. 
This is a property of signed networks that corresponds to the possibility of exactly partitioning the signed graph into two adversary sub-communities, such that all links within each sub-community have positive weight, whereas all links between nodes of different communities have negative weights.  
When the graph is structurally balanced, the {\black SNG} can be transformed through a change of variables into a network coordination game. Consensus \tcb{action profile}s in the network coordination game corresponds to \tcb{action profile}s that take opposite sign in the two communities of the structurally balanced graph. 

First, to every $\sigma$ in $\mc X=\{\pm1\}^{\mc V}$, we can associate a diagonal matrix $[\sigma]$ in $\R^{\mc V\times\mc V}$ with diagonal entries \be\label{sigma}[\sigma]_{ii}=\sigma_i\,,\qquad \forall i\in\mc V\,.\ee
Such matrices identify linear operators in $\R^n$ that are referred to as \textit{gauge transformations} in some of the literature \cite{Altafini:2013}. 

\begin{definition}
For a network $\mc G=(\mc V,\mc E,W)$, and a vector $h$ in $\R^{\mc V}$, consider the {\black SNG} with binary actions on $\mc G$ with  external field $h$. Given $\sigma$ in $\mc X$, let the $[\sigma]$-transformed network be 
\be\label{sigmaG}\mc G^{[\sigma]}=(\mc V,\mc E,W^{[\sigma]})\,,\qquad W^{[\sigma]}=[\sigma]W[\sigma]\,,\ee
and the $[\sigma]$-transformed external field \be\label{sigmad}h^{[\sigma]}=[\sigma]h\,.\ee
Then, the \emph{$[\sigma]$-transformed} game is the {\black SNG} with binary actions on $[\sigma]$-transformed network  $\mc G^{[\sigma]}$  with $[\sigma]$-transformed external field $h^{[\sigma]}$. 
\end{definition}
Observe that the utility function of every player $i$ in $\mc V$ in the $[\sigma]$-transformed game defined above is given by  
\tcb{\be\label{sigmau}\ba{rcl}\ds u^{[\sigma]}_i(x)
&=&\ds h_i^{[\sigma]}x_i+ x_i\sum_{j\in\mc V}(W^{[\sigma]})_{ij}x_j\\[5pt] 
&=&\ds h_i\sigma_ix_i+ x_i\sigma_i\sum_{j\in\mc V}W_{ij}\sigma_jx_j\\[5pt] 
&=&\ds h_i([\sigma]x)_i+ ([\sigma]x)_i\sum_{j\in\mc V}\sigma_iW_{ij}([\sigma]x)_j\\[5pt] 
&=&\ds  u_i([\sigma]x)\,,\ea\ee}
for every \tcb{action profile} $x$ in $\mc X$, where the first and the last identities follow from \eqref{eq:ut}, the second one from \eqref{sigmaG} and \eqref{sigmad}, and the third one from \eqref{sigma}. 
Equation \eqref{sigmau} shows that, if we apply a gauge transformation $[\sigma]$ to a {\black SNG}, we obtain a new {\black SNG} whose utility functions coincide with the ones of the original game evaluated in the transformed \tcb{action profile}s $[\sigma]x$.
This immediately leads to the following result. 
\begin{lemma}\label{lemma:NashGauge} Consider a {\black SNG} with binary actions on a network $\mc G=(\mc V,\mc E,W)$ with external field $h$, and let $\mc N$ be the set of its Nash equilibria. Given $\sigma$ in $\mc X$, consider the $[\sigma]$-transformed  game and let $\mc N_{[\sigma]}$ be the set of its Nash equilibria. Then, 
\be\label{Nash=Nash}\mc N=\{[\sigma]x^*:\,x^*\in\mc N_{[\sigma]}\}\,.\ee
\end{lemma}
\begin{IEEEproof}
See Appendix \ref{sec:proof-lemma-NashGauge}.\end{IEEEproof}\medskip


 Gauge transformations are particularly useful when the graph is structurally balanced, as per the following definition.
\begin{definition}\label{sb} Consider a network $\mc G=(\mc V,\mc E, W)$. Then: 
\begin{enumerate}
\item[(i)] a \emph{balanced partition} of $\mc G$ is a binary partition 
\be\label{balanced-partition} \mc V_1\cup \mc V_2=\mc V\,,\qquad \mc V_1\cap \mc V_2=\emptyset\,,\ee
of its node set $\mc V$ such that 
\be\label{balanced-partition-1} W_{ij}\geq 0\,,\qquad \forall i,j \in \mc V_q\,,\ q=1,2\,,\ee
and
\be\label{balanced-partition-2} W_{ij}\leq 0\,,\qquad \forall i \in \mc V_q\,,\ j\in\mc V_r\,,\ q\ne r\,,\ q,r\in\{1,2\}\,;\ee
\item[(ii)] $\mc G$ is \emph{structurally balanced} if it admits a balanced partition. 
\end{enumerate}
\end{definition}

Definition \ref{sb} says that a network is structurally balanced if its node set can be split into two opposite parts such that links connecting nodes in the same part all have nonnegative weight, while links connecting nodes in opposite parts all have nonpositive weight. We have the following simple result. 
\begin{lemma} \label{lemma:balanced}
A network $\mc G=(\mc V,\mc E, W)$ is structurally balanced if and only if there exists a gauge transformation $[\sigma]$ such that the $[\sigma]$-transformed network $\mc G^{[\sigma]}$ is unsigned. 
\end{lemma}
\begin{IEEEproof} This result coincides with \cite[Lemma 1]{Altafini:2013}. However, since the work \cite{Altafini:2013} is developed under the ``digon sign-symmetry'' assumption (c.f.~footnote 1), and we are working in greater generality, we present a proof in Appendix \ref{sec:proof-lemma-balanced}.   
\end{IEEEproof}\medskip

Lemma \ref{lemma:balanced} and Proposition \ref{prop:super-modular} together imply the following result. 

\begin{proposition}\label{prop:structurally-balanced-equilibria}
Consider a {\black SNG} with binary actions on a network $\mc G=(\mc V,\mc E,W)$ with external field $h$. 
If $\mc G$ is structurally balanced, then the set of Nash equilibria $\mc N$ is  globally BR-reachable. \end{proposition}
\begin{IEEEproof}
If $\mc G$ is structurally balanced, Lemma \ref{lemma:balanced} implies that there exists a gauge transformation $[\sigma]$ such that the whole node set $\mc V$ is coordinating in the $[\sigma]$-transformed network $\mc G^{[\sigma]}$. It then follows from Proposition \ref{prop:super-modular}(ii) that the $[\sigma]$-transformed {\black SNG} with binary action on the network $\mc G^{[\sigma]}$ with external field $h^{[\sigma]}$  has a non-empty and globally reachable set of Nash equilibria $\mc N_{[\sigma]}$. 
Lemma \ref{lemma:NashGauge} then implies that the set of Nash equilibria of the {\black SNG} with binary actions on $\mc G$ with external field $h$ satisfies  $\mc N=\{[\sigma]x^*:\,x^*\in\mc N_{[\sigma]}\}\ne\emptyset$ and is globally reachable. 
\end{IEEEproof}

\begin{example}\label{ex:sb}
	Consider the graph in Figure \ref{fig:sb}. The graph is not structurally balance\tcb{d}. Anyway, the subset $\mc R=\{1, \dots, 4\}$ is such that the graph $\mc G_{\mc R}$ is structurally balanced. Indeed, the partition $\mc R = \mc R_1 \cup \mc R_2$ with $\mc R_1=\{1,4\}$ and $\mc R_{\black 2}=\{2,3\}$ is such that $W_{\mc R_q \mc R_q}\geq 0$ for $q$ in $\{1,2\}$ and $W_{\mc R_q \mc R_r}\leq 0$ for $q\neq r$, $q$, $r$ in $\{1,2\}$. According to \tcb{Lemma  \ref{lemma:balanced}}, there exists a gauge transformation $[\sigma]$ such that $\mc R$ is a coordinating set of $\mc G^{[\sigma]}$. Indeed, \tcb{the gauge transformation $[\sigma]$ with $\sigma=(1,-1,-1,1,1,1)$ is such  that } $W^{[\sigma]}=[\sigma]W[\sigma]$ is such that
		$$
			\tilde W_{\mc R\mc R}= \tcb{[\sigma_{\mc R}]}\left[\begin{matrix}
				0 &-2&0&2\\
				-2 &0&1&0\\
				-1 &1&0&-2\\
				2&-1&0&0
			\end{matrix}\right]  \tcb{[\sigma_{\mc R}]}
			=\left[\begin{matrix}
				0 &2&0&2\\
				2 &0&1&0\\
				1 &1&0&2\\
				2&1&0&0
			\end{matrix}\right]\,,
		$$
		\tcb{that is a nonnegative matrix.}
\end{example}

{\color{black}\begin{figure}
		\centering
		\begin{tikzpicture}
			\foreach \x/\name/\value in {(0,0)/4/1, (2,2)/2/1, (2,0)/3/1, (0,2)/1/1}\node[shape=circle,draw, , fill=gray!30!white](\name) at \x {\small \name};
			\path [<->,draw] (2) edge node[right] {\small 1} (3);
			\path [<->,draw, red] (2) edge node[above] {\small -2} (1);
			\path [<->,draw] (1) edge node[left] {\small 2} (4);
			\path [->,draw,red] (3) edge node[below] {\small -2} (4);
			\path [->,draw,red] (3) edge node[above, near start] {\small -1} (1);
			\path [->,draw,red] (4) edge node[above, near start] {\small -1} (2);
			
			\node[shape=circle, draw ](5) at (-1.5,1) {\small 5};
			\node[shape=circle, draw ](6) at (3.5,1) {\small 6};
			
			\path[<->, draw] (4) edge node[above] {\small 1} (5);
			\path[<->, draw, red] (1) edge node[above] {\small -1} (5);
			\path[->, draw] (3) edge node[above] {\small 1} (6);
			\path[<->, draw,red] (2) edge node[above] {\small -1} (6);
			
			
		\end{tikzpicture}	
		\caption{Graph studied in Example \ref{ex:sb}. The subset $\mc R=\{1, \dots, 4\}$ (in gray) is such that $\mc G_{\mc R}$ is structurally balanced.}
		\label{fig:sb}
\end{figure}
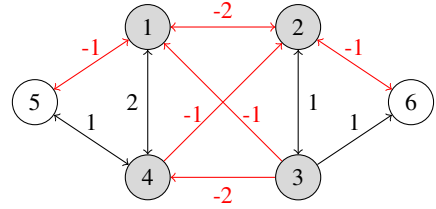}

\section{Main results }\label{sec:results}
In this section, we present our main results. 

Before proceeding, we introduce some notation that will play a crucial role in the following.
For every node $i$ in $\mc V$ and subset of nodes $\mc B\subseteq\mc V$, we define
\be\label{B-outdegree}w_i^{\mc B} = \sum_{j\in \mc B} {\color{black}|W_{ij}|}\,,\ee
to be its $\mc B$-\textit{out-degree}. 
In the special case when $\mc B=\mc V$ coincides with the whole node set, we simply refer to $$w_i=w_i^{\mc V}\,,$$ as the \textit{out-degree} of a node $i$ in $\mc V$ and let $w=|W|\1$ be the vector of out-degrees. 

We will often consider binary partitions of $\mc V$ of the type 
\be\label{binary-partition}\mc V={\mc R}\cup{\mc S}\,,\qquad \mc R\cap\mc S=\emptyset\,,\ee 
and identify the \tcb{action profile} space as  the Cartesian product 
$$\mc X=\mc X_{\mc R}\times \mc X_{\mc S}\,,$$ where $\mc X_{\mc R}=\mc A^{{\mc R}}$ and $\mc X_{\mc S}=\mc A^{{\mc S}}$. Correspondingly, we decompose every \tcb{action profile} $x$ in $\mc X$ as 
$$x=(x_{\mc R},x_{\mc S})\,,$$ where $x_{\mc R}$ in $\mc X_{\mc R}$ is the \tcb{action profile} of the players in $\mc R$ and $x_{\mc S}$ in $\mc X_{\mc S}$ is the \tcb{action profile} of the  players in $\mc S$. 

For every \tcb{action profile} $z$ in $\mc X_{\mc S}$, we consider the 
{\black SNG} with binary actions on the subnetwork $\mc G_{\mc R}$ with external field 
$$h^{(z)}=h_{\mc R}+W_{\mc R\mc S}z\in \R^{\mc R}\,,$$ 
i.e., the game with player set ${\mc R}$, action set $\mc A=\{\pm1\}$, and utility functions
\tcb{\be\label{eq:perturbeh_coord}
	\ba{rcl}u_i^{(z)}(y) 
	&=&\ds h^{(z)}_iy_i+y_i \sum_{j \in \mc S} W_{ij}y_j\\
	&=&\ds h_iy_i + y_i \sum_{j \in \mc S} W_{ij}z_j+ y_i \sum_{j \in \cp} W_{ij} y_j \,,\ea
\ee}
for every player $i$ in $\cp$ and \tcb{action profile} $y$ in $\mc X_{\mc R}$. We shall refer to this game as the \emph{$\mc R$-restricted} {\black SNG} with \tcb{action profile} of players in $\mc S$ frozen to $z$.

Analogously, for a given \tcb{action profile} $y$ in $\mc X_{\mc R}$, the $\mc S$-restricted game with \tcb{action profile} of players in $\mc R$ frozen to $y$ refers to the game with player set $\mc S$,  action set $\mc A=\{\pm1\}$,  and utility functions
\tcb{\be\label{eq:S-utilities}
	u_i^{(y)}(z) = z_ih_i+z_i \sum_{j \in \cp} W_{ij} y_j+z_i \sum_{j \in \mc S} W_{ij}z_j\,,
\ee}
for every player $i$ in $\mc S$ and \tcb{action profile} $z$ in $\mc X_{\mc S}$.
This can be interpreted as the {\black SNG} with binary actions on the subnetwork $\mc G_{\mc S}$ with external field 
$$h^{(y)}=h_{\mc S}+W_{\mc S\mc R}y\in \R^{\mc R}\,. $$ 

We denote by $\mc N_{\mc R}^{(z)}$ the set of Nash equilibria of the $\mc R$-restricted {\black SNG} with \tcb{action profile} of players in $\mc S$ frozen to $z$ and by $\mc N_{\mc S}^{(y)}$ the set of Nash equilibria of the $\mc S$-restricted {\black SNG} with \tcb{action profile} of players in $\mc R$ frozen to $y$.

\subsection{Existence of Nash equilibria}\label{sec:existence}
In this subsection, we investigate the existence of Nash equilibria of {\black SNG}s with binary actions.  

First, we consider a binary partition of the set of players as in \eqref{binary-partition}, such that the subnetwork $\mc G_{\mc R}$ is unsigned, and we look for Nash equilibria of the {\black SNG} whose projection on the coordinating set $\mc R$ 
is a consensus \tcb{action profile}, i.e., $$x^*=(a\1, z^*)\,,$$ for some $a=\pm 1$. 
Our first result, stated below, guarantees the existence of such Nash equilibria in terms of two properties of the $\mc R$- and $\mc S$-restricted {\black SNG}s, \tcb{respectively}: the first one, cohesiveness, limits the influence that players in $\mc S$ can have on the players in $\mc R$, the second one ensures that the $\mc S$-restricted {\black SNG} with \tcb{action profile} of players in $\mc R$ frozen to a consensus admits a Nash equilibrium.

\begin{proposition}\label{prop:mixed1}
Consider a {\black SNG} with binary actions on a network $\mc G=(\mc V,\mc E,W)$ with external field $h$,  and a binary partition \eqref{binary-partition} such that $\mc G_{\mc R}$ is unsigned. Assume that 
	there exists an action $a$ in $\{\pm1\}$ such that 
			\be\label{assumptioni}w^{{\mc R}}_i +ah_i\geq w^{{\mc S}}_i\,,\qquad \forall i\in\mc R\,,\ee
and  that \be\label{NSnon-empty}\mc N_{{\mc S}}^{(a\1)}\ne\emptyset\,.\ee
	Then, there exists a Nash equilibrium $x^*$ in $\mc N$ such that \be\label{x*R=a1}x^*_{\mc R}=a\1\,.\ee 
\end{proposition}
\begin{IEEEproof}         
Let $x^*$ in $\mc X$ be a \tcb{action profile} such that $x^*_{\mc R}=a\1$ and $x^*_{\mc S}=z^*$ in $\mc N_{{\mc S}}^{(a\1)}$ is a Nash equilibrium of the $\mc S$-restricted {\black SNG} with \tcb{action profile} of players in $\mc S$ frozen to $a\1$. 
We will now prove that $x^*$ is a Nash equilibrium of the {\black SNG}. 
	Since $z^*\in\mc N_{{\mc S}}^{(a\1)}$, what we are left to show is that $a\in \mc B_i(x_{-i}^*)$ for every player $i$ in $\mc R$. 
	 Since $W_{ij}\geq 0$ for every player $i$ and $j$ in $\mc R$, 
from  \eqref{eq:perturbeh_coord} this is equivalent to 
	$$
	a\left(aw_i^{{\mc R}}+h_i+(W_{{\mc R}{\mc S}}z^*)_i\right)\geq -a\left(aw_i^{{\mc R}}+h_i+(W_{{\mc R}{\mc S}}z^*)_i\right)\,,$$
	or, equivalently, that
	\be\label{wiR}w_i^{{\mc R}}+ah_i\geq -a(W_{{\mc R}{\mc S}}z^*)_i\,,\qquad \forall i\in\mc R\,.\ee
	Since $$w^{{\mc S}}_i\geq -a(W_{{\mc R}{\mc S}}z^*)_i\,,\qquad  \forall i\in\mc R\,,$$ \eqref{wiR} follows from \eqref{assumptioni}. This completes the proof. 
\end{IEEEproof}\medskip
\tcb{Proposition \ref{prop:mixed1} provides sufficient conditions for the existence of Nash equilibria when $\mc R$ is a set of coordinating players. The cohesiveness condition \eqref{assumptioni} guarantees that, when all players in $\mc R$ coordinate on action $a$, the combined effect of the external field and the interactions within $\mc R$ dominates, for every player in $\mc R$, the maximal adverse influence exerted by players in $\mc S$, namely $w_i^{\mc S}$. Throughout, we refer to this condition as a cohesiveness property, since it generalizes the original notion introduced in \cite{Morris:2000}, as discussed in Remark \ref{remark:cohesive}. Condition \eqref{NSnon-empty}, in turn, requires that the $\mc S$-restricted game, with the actions of players in $\mc R$ fixed at $a\1$, admits at least one Nash equilibrium. Together, these conditions ensure the existence of a Nash equilibrium of the original game satisfying $x^*_{\mc R}=a\1$.  }
\begin{remark} \label{remark:cohesive}
	In the special case when $h=0$, i.e., when all players are unbiased, using the identity $w_i^{{\mc S}}=w_i-w_i^{{\mc R}}$, we can rewrite condition \eqref{assumptioni} in Proposition \ref{prop:mixed1} as $$w_i^{{\mc R}}\geq 1/2w_i\,,\qquad \forall i\in{\mc R}\,.$$ 
	This says that every node in ${\mc R}$ has at least half of the weight of its out-links towards other nodes in ${\mc R}$. In the terminology introduced in \cite{Morris:2000}, this says that ${\mc R}$ is a $1/2$-cohesive subset of $\mc V$ relative to the unsigned graph $(\mc V,\mc E, |W|)$.
\end{remark}

\begin{remark}
Notice that the two networks of Examples \ref{ex:discoord} and \ref{ex:3ringanticoord}, for which Nash equilibria did not exist, were missing just one of the two assumptions in Proposition \ref{prop:mixed1}. Precisely, condition \eqref{assumptioni} was not satisfied by the network in Figure \ref{fig:disc_game} (a), while condition \eqref{NSnon-empty} was not satisfied by the network in Figure \ref{fig:disc_game} (b). We also observe that global reachability is not guaranteed by conditions \eqref{assumptioni} and  \eqref{NSnon-empty}  alone, as already noticed in Example \ref{ex:nostability}. 
\end{remark}
We now present the following result extending Proposition \ref{prop:mixed1} to cases when $\mc G_{\mc R}$ is structurally balanced rather than unsigned. In this case, provided that $\mc R$ remains cohesive, we can determine sufficient conditions for the existence of Nash equilibria $x^*$ of the {\black SNG} whose restriction $\tau=x^*_{\mc R}$ to the set $\mc R$ is such that the $[\tau]$-transformed subnetwork $\mc G_{\mc R}^{[\tau]}$ is unsigned.

\begin{theorem}\label{theo:main-existence}
	Consider a {\black SNG} with binary actions on a network $\mc G=(\mc V,\mc E,W)$ with external field $h$. 
	Assume that there exists a binary partition as in \eqref{binary-partition} such that $\mc G_{\mc R}$ is structurally balanced  
	and let $\tau$ in $\mc X_{\mc R}$ be such that 
	\be\label{tauWtau}\tau_iW_{ij}\tau_j\ge0\,,\qquad \forall i,j\in\mc R\,.\ee 
	If 
		\be\label{wRitau}
		w^{{\mc R}}_i + \tau_ih_i\geq w^{{\mc S}}_i\,,\qquad \forall i\in\mc R\,,
		\ee
		and
		\be\label{NStau}\mc N_{{\mc S}}^{(\tau)}\ne\emptyset\,,\ee
	then, there exists a Nash equilibrium $x^*$ in $\mc N$  such that 
	\be\label{xRtau}x_{\mc R}^*=\tau\,.\ee
\end{theorem}
\begin{IEEEproof}
	Let $\sigma$ in $\mc X$ be such that $\sigma_{\mc R}=\tau$ and $\sigma_{\mc S}=\1$. 
Consider the  associated gauge transformation $[\sigma]$ and the {\black SNG} with binary actions on the $[\sigma]$-transformed network  $\mc G^{[\sigma]}$ with external field $h^{[\sigma]}$. 
Observe that condition \eqref{tauWtau} ensures that $\mc G^{[\sigma]}_{\mc R}$ is unsigned.  
Moreover, assumption \eqref{wRitau} implies that 
	\be\label{cond1}w_i^\mc R+h^{[\sigma]}_i \geq w_i^{\mc S},\qquad \forall i\in\mc R\,,\ee
	where $w_i^{\mc R}$ and $w_i^{\mc S}$ denote, respectively, the $\mc R$- and $\mc S$-out-degrees of node $i$ (as defined in \eqref{B-outdegree}) that are the same in both the original network $\mc G$ and the $[\sigma]$-transformed one $\mc G^{[\sigma]}$, since gauge transformations do not alter degrees.
Furthermore, it follows from equation \eqref{sigmau} that 	
	$$u_i^{[\sigma]}(\1,  z)=u_i([\sigma](\1, z))=u_i(\tau, z)\,,\qquad \forall i\in\mc S\,,$$
	which, together with assumption \eqref{NStau},  implies that 
	$$(\mc N_{[\sigma]})_{{\mc S}}^{(\tau)}=\mc N_{{\mc S}}^{(\tau)}\ne\emptyset\,.$$
	Therefore, the $[\sigma]$-transformed {\black SNG} 
satisfies all the assumptions of Proposition \ref{prop:mixed1}, hence it admits a Nash equilibrium $\tilde x^*$ in $\mc N_{[\sigma]}$ such that $\tilde x^*_{\mc R}=\1$. 
Lemma \ref{lemma:NashGauge} then implies that $x^*=[\sigma]\tilde x^*\in\mc N$ is a Nash equilibrium for the {\black SNG} with binary actions on $\mc G$ with external field $h$. Finally, observe  that $$x^*_{\mc R}=([\sigma]\tilde x^*)_{\mc R}=[\tau]\1=\tau\,,$$
so that \eqref{xRtau} holds true, thus completing the proof. 
\end{IEEEproof}\medskip
\tcb{Theorem \ref{theo:main-existence} generalizes the previous existence result from a coordinating set to the broader setting of a structurally balanced subgraph $\mc R$. In this case, the vector $\tau$ identifies the gauge transformation that renders the subgraph $\mc G_{\mc R}$ unsigned, as characterized by condition \eqref{tauWtau}. Accordingly, $\tau$ also represents the polarized configuration naturally associated with the structural balance of $\mc G_{\mc R}$. Condition \eqref{wRitau} then requires that, when players in $\mc R$ adopt the polarized profile $\tau$, the combined effect of the external field and the aligned interactions within $\mc R$ is sufficiently strong so that the influence exerted by players in $\mc S$ cannot prevent such a configuration from being a Nash equilibrium. Condition \eqref{NStau} is the natural counterpart of \eqref{NSnon-empty}, with the actions on $\mc R$ fixed to the polarized profile $\tau$ instead of the coordinated profile $a\1$. In this case, the resulting Nash equilibrium exhibits the polarized configuration $\tau$ on the set $\mc R$, namely, $x^*_{\mc R}=\tau$.} Observe that assumption \eqref{NStau} in  Theorem \ref{theo:main-existence}, namely  the fact that the set $\mc N_{{\mc S}}^{(\tau)}$ of Nash equilibria of the $\mc S$-restricted {\black SNG} with \tcb{action profile} of players in $\mc R$ frozen to $\tau$, is automatically verified when the subnetwork $\mc G_{\mc S}$ is either undirected or structurally balanced itself. This leads to the following corollaries.

\begin{corollary}\label{th:mixed_existence1}
	Consider a {\black SNG} with binary actions on $\mc G=(\mc V,\mc E,W)$ and a binary partition \eqref{binary-partition} where both  $\mc G_{\mc R}$ and $\mc G_{\mc S}$ are structurally balanced. If 
	there exists $\tau$ in $\mc X_{\mc R}$ such that assumption \eqref{wRitau} holds true, then there exists a Nash equilibrium $x^* \in \mc N$ satisfying equation \eqref{xRtau}. 
\end{corollary}
\begin{IEEEproof}
	If the subnetwork $\mc G_{\mc S}$ is structurally balanced, then, for every $\tau$ in $\mc X_{\mc R}$, Proposition \ref{prop:structurally-balanced-equilibria} applied to  the $\mc S$-restricted {\black SNG} with \tcb{action profile} of players in $\mc R$ frozen to $\tau$, implies that assumption \eqref{NStau} holds true. 
	The claim then follows from Theorem \ref{theo:main-existence}. 
\end{IEEEproof}\medskip

\begin{corollary}\label{th:mixed_existence2}
	Consider a {\black SNG}  with binary actions on $\mc G=(\mc V,\mc E,W)$ and a binary partition \eqref{binary-partition} where $\mc G_{\mc R}$ is structurally balanced and $\mc G_{\mc S}$ is undirected. If 
	there exists $\tau$ in $\mc X_{\mc R}$ such that condition \eqref{wRitau} holds true, then there exists a Nash equilibrium $x^* \in \mc N$ satisfying equation \eqref{xRtau}. 
\end{corollary}
\begin{IEEEproof}
	If the subnetwork $\mc G_{\mc S}$ is undirected,  then, for every $\tau$ in $\mc X_{\mc R}$, Proposition \ref{pr:pot_mixed} applied to  the $\mc S$-restricted {\black SNG} with \tcb{action profile} of players in $\mc R$ frozen to $\tau$  implies that assumption \eqref{NStau} holds true. 
	The claim then follows from Theorem \ref{theo:main-existence}. 
%
%
\end{IEEEproof}\medskip
\begin{remark}	Corollaries \ref{th:mixed_existence1} and \ref{th:mixed_existence2}  significantly generalize previous results where the existence of (pure strategy) Nash equilibria was proved only for network coordination or network anti-coordination games over undirected graphs \cite{Ramazi.Riehl.Cao:2016}. Moreover, Corollary \ref{th:mixed_existence2}  applies to the mixed network coordination/anti-coordination games studied in \cite{vanelli2020games, arditti2021equilibria}.\end{remark}

\addtocounter{example}{-6}
\begin{example}[cont'd]
	Consider the graph in Figure \ref{fig:graph1} and let $h=0$. 
	Observe that, for all $i$ in $\mc R$, it holds that
	$w_i^{\mc R}\geq w_i^{\mc S}$. Since the subnetwork $\mc G_{\mc R}$ is unsigned and the \tcb{subnetwork} $\mc G_{\mc S}$ is undirected, 
	Corollary \ref{th:mixed_existence2} implies the existence of two Nash equilibria $x^*$ and $x^{**}$ in $\mc N$ such that $x^*_{\mc R}=\1=-x^{**}_{\mc R}$. 
\end{example}
\addtocounter{example}{5}


\subsection{Stability of Nash equilibria}

We now present results on the global stability of Nash equilibria of {\black SNG}s. Before proceeding, it is convenient to reconsider the {\black SNG} in Example \ref{ex:nostability}, whose set of Nash equilibria was shown to be not globally BR-reachable. A closer look at this example suggests that this is a direct consequence of the topological structure of the underlying network displayed in Figure \ref{fig:noeq_nostability}. 
This network contains two components $\{1,2,3\}$ and $\{5,6,7\}$, each of which without any out-link towards other nodes in the graph. This decomposition of the graph is what prevents the coordinating players to reach a consensus starting from a polarized initial condition. 
The observation above motivates the following definition that introduces a property of the graph guaranteeing that similar decompositions are not possible. The proposed definition also accounts for the external field $h$. It will be at the basis of the results presented in this subsection.
 \begin{definition}\label{def:indecomp}
 Consider a network $\mc G= (\mc V, \mc E, W) $ and two vectors  $h^-$ and $h^+$ in $\R^\mc V$ such that $h^{-}\leq h^+$. We say that $\mc G$ is  $(h^-,h^+)$-indecomposable if for every binary partition 
 \be\label{binarypartition+-}\mc V = \mc V^-\cup \mc V^+\,,\qquad \mc V^-\cap\mc V^+=\emptyset\,,\qquad\mc V^-\ne\emptyset\ne\mc V^+\,,\ee 
 there exists a node $i$ in $\mc V$ such that either
\begin{equation}\label{eq:ind1}
	i \in \mc V^+ \text{ and } w_i^{\mc V^+} + h_i^+<  w_i^{\mc V^-}
	\end{equation}
	or 
	\begin{equation}\label{eq:ind2} 
		i \in \mc V^- \text{ and } w_i^{\mc V^-}-h_i^- < w_i^{\mc V^+}\,.
	\end{equation}
\end{definition}

The following example illustrates the notion of indecomposability introduced in Definition \ref{def:indecomp} above. 

\begin{figure}
	\centering
	\begin{tikzpicture}
		\hspace{0.2cm}
		\foreach \x/\name/\value in {(0,0)/4/1, (2,2)/2/1, (2,0)/3/1, (0,2)/1/1}\node[shape=circle,draw](\name) at \x {\small \name};
		\path [<->,draw] (2) edge node[right] {\small 1} (3);
		\path [<->,draw] (2) edge node[above] {\small 2} (1);
		\path [<->,draw] (1) edge node[left] {\small 2} (4);
		\path [->,draw] (3) edge node[below] {\small 1} (4);
		\path [->,draw] (3) edge node[above, near start] {\small 1} (1);
		\path [->,draw] (4) edge node[above, near start] {\small 1} (2);

	\end{tikzpicture}	
	\caption{Graph considered in Example \ref{ex:dec1}.}
	\label{fig:indec}
\end{figure}
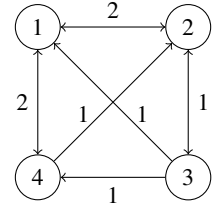

\begin{example}\label{ex:dec1}
	The network $\mc G$ in Figure \ref{fig:indec} is $(h^{-}, h^{+})$-indecomposable for $h^{+}=(3,2,0,2)=-h^{-} $. To see this, first notice that, for any \tcb{partition} where $|\mc V^+|=|\mc V^-|=2$, node $3$ satifies either \eqref{eq:ind1} or \eqref{eq:ind2} as a direct consequence of the facts that its out-degree is $w_3=3$ and $h_3^-=h_3^+=0$. 
Suppose instead that $|\mc V^-	|=1$. A direct check shows that the node in $\mc V^-$ satisfies \eqref{eq:ind2}. The case when $|\mc V^+|=1$ is identical. 

	The network $\mc G$ in Figure \ref{fig:indec} is not $(h^{-}, h^{+})$-indecomposable for $h^{+}=\1$ and  $h^{-} = -h^{+}$. Indeed, if we take the partition $\mc V^+=\{1,4\}$ and $\mc V^-=\{2,3\}$, we have that $w_1^{\mc V^+}+h_1^+=3\geq 2= w_1^{\mc V^-}$, $w_4^{\mc V^+}+h_4^+=3\geq 1 = w_4^{\mc V^-}$, $w_2^{\mc V^+}-h_2^-=2\geq 1= w_2^{\mc V^+}$ and $w_3^{\mc V^-}-h_3^-=2\geq 2=w_3^{\mc V^+}$. Then, $\nexists i$ in $\mc V^+$ or $\mc V^-$ such that either \eqref{eq:ind1} or \eqref{eq:ind2} is violated.
\end{example}

Consider a {\black SNG} with binary actions on a network $\mc G=(\mc V,\mc E,W)$ with external field $h$ such that 
\be\label{h-bounds} h^-_i\le h_i\le h_i^+\,,\qquad\forall i\in\mc V\,.\ee
Given a binary partition as in \eqref{binarypartition+-}, consider the \tcb{action profile} $x$ in $\mc X$ with  $$x_{\mc V^+}=\1\,,\qquad x_{\mc V^-}=-\1\,.$$ 
Then, conditions \eqref{eq:ind1} and \eqref{eq:ind2} imply that there exists a player $i$ in either  $\mc V^+$ or $\mc V^-$  that is not playing best response in \tcb{action profile} $x$. This implies that the {\black SNG} on $\mc G$ with external field $h$  admits no coexistent Nash equilibria, i.e., no Nash equilibria other than, possibly, consensus \tcb{action profile}s. Notice that the absence of coexistent Nash equilibria implied by the $(h^-,h^+)$-indecomposability of the graph is robust with respect to changes of the vector $h$ in the hyper-rectangle $\{h\in\R^{\mc V}:\,\eqref{h-bounds}\}$.

We will make use of a result in \cite{arditti2024robust} that ensures, for network coordination games on $(h^-,h^+)$-indecomposable unsigned networks, the existence of a BR-path from every \tcb{action profile} $x$ to a consensus \tcb{action profile} that is independent from the specific choice of the vector $h$ satisfying \eqref{h-bounds}. Precisely, \cite[Theorem 4(i)]{arditti2024robust} implies the following.

\begin{lemma}\label{lemma-robust}  Consider an unsigned network $\mc G= (\mc V, \mc E, W) $ and two vectors  $h^-$, $h^+$ in $\R^\mc V$ for which $\mc G$ is $(h^-,h^+)$-indecomposable. Then, for every \tcb{action profile} $x^{(0)}$ in $\mc X$, there exists an $l$-tuple of \tcb{action profile}s $(x^{(1)},\ldots,x^{(l)})$,  with $1\le l\leq n$, such that $x^{(l)}\in\{\pm\1\}$ is a consensus  profile, and $(x^{(0)},x^{(1)},\ldots,x^{(l)})$ is a BR-path for every {\black SNG} with binary actions on $\mc G$ with external field $h$ satisfying \eqref{h-bounds}. 

\end{lemma}

We can now get the following result.

\begin{proposition}\label{prop:stability}
Consider a {\black SNG} with binary actions on a network $\mc G=(\mc V,\mc E,W)$ with external field $h$ and a binary partition \eqref{binary-partition} such that $\mc G_{\mc R}$ is unsigned. Let $h^-$ and $h^+$ in $\R^{\mc R}$ be the vectors with entries 
\be\label{h+-}h^+_i=h_i+w_i^{{\mc S}}\,, \qquad h^-_i=h_i-w_i^{{\mc S}}\,,\qquad\forall i\in\mc R\,.\ee
Assume that $\mc G_{\mc R}$ is $(h^-,h^+)$-indecomposable and 
\be\label{wRi}w^{{\mc R}}_i -|h_i|> w^{{\mc S}}_i\,,\qquad\forall i\in\mc R\,.\ee
Then:  
\begin{enumerate}
\item[(i)] if, for $y=\pm\1$ in $\mc X_{\mc R}$, the set 
$\mc N_{{\mc S}}^{(y)}$  is globally BR-reachable for the $\mc S$-restricted {\black SNG} with action profile of players in $\mc R$ frozen to $y$, 
then, the subset  of Nash equilibria \be\label{tildeN}\tilde{\mc N}=\{x^*\in\mc N:\,x^*_{\mc R}\in\{\pm\1\}\}\,,\ee  is non-empty and globally BR-reachable. 
	\end{enumerate}
	Moreover: 
	\begin{enumerate}
	\item[(ii)] if, for $y=\pm\1$ in $\mc X_{\mc R}$, there exists a non-empty subset $\bar{\mc N}_{{\mc S}}^{(y)}\subseteq \mc N_{{\mc S}}^{(y)}$ that is globally BR-stable for the $\mc S$-restricted {\black SNG} with action profile of players in $\mc R$ frozen to $y$, then there exists a non-empty globally BR-stable subset $\ov{\mc N}\subseteq\tilde{\mc N}$. 
\end{enumerate}
\end{proposition}

\begin{IEEEproof} 
Fix an arbitrary \tcb{action profile} $x$ in $\mc X$ and let $z=x_{\mc S}$. 
On the one hand, since the subnetwork $\mc G_{\mc R}$ is unsigned and $(h^-,h^+)$-indecomposable, Lemma \ref{lemma-robust} implies that the set of consensus \tcb{action profile}s $\{\pm\1\}\subseteq\mc X_{\mc R}$ is globally BR-reachable for the $\mc R$-restricted network coordination game with action profile of players in $\mc S$ frozen to $z$, so that there exists a length-$l$ BR-path $((y^{(0)},z),(y^{(1)},z),\ldots,(y^{(l)},z))$ with $y^{(0)}=x_{\mc R}$ and $y^{(l)}=a\1$, for some $a$ in $\{\pm1\}$. 
	On the other hand, since the set $\mc N_{{\mc S}}^{(a\1)}$  is globally BR-reachable for the $\mc S$-restricted {\black SNG} with action profile of players in $\mc R$ frozen to $y=a\1$, there exists a length-$m$ BR-path 
	$((a\1,z^{(0)}),(a\1,z^{(1)}),\ldots,(a\1,z^{(m)}))$ with $z^{(0)}=z$ and $z^{(m)}=z^*\in\mc N_{{\mc S}}^{(a\1)}$. 
	
	Observe that the \tcb{action profile} $x^*$ in $\mc X$ with $x^*_{\mc R}=a\1$ and  $x^*_{\mc S}=z^*$ is a Nash equilibrium for the original {\black SNG} with binary actions on $\mc G$ with external field $h$, as every player $i$ in $\mc R$ is playing best response thanks to \eqref{wRi}, while every player $j$ in $\mc S$ is playing best response since $z^*\in\mc N_{{\mc S}}^{(a\1)}$. We have thus found a length-$(l+m)$ BR-path 
	$((y,z),(y^{(1)},z),\ldots,(a\1,z),(a\1,z^{(1)}),\ldots,(a\1,z^*))$ from $x$ to $x^*$ in $\mc N$ with $x^*_{\mc R}=a\1$. 
	The arbitrariness of initial \tcb{action profile} $x$ in $\mc X$ implies that the set $\tilde{\mc N}$ defined in \eqref{tildeN}, i.e., the subset of Nash equilibria in which players in $\mc R$ are at consensus is globally BR-reachable, thus proving point (i) of the claim. 
	
	To prove point (ii) of the claim, observe that, if there exists a non-empty subset $\bar{\mc N}_{{\mc S}}^{(a\1)}\subseteq \mc N_{{\mc S}}^{(a\1)}$ that is globally BR-stable for the $\mc S$-restricted {\black SNG} with \tcb{action profile} of players in $\mc R$ frozen to $y=a\1$, then the BR-path above can be constructed leading to $x^*$ such that $x^*_{\mc S}\in\bar{\mc N}_{{\mc S}}^{(a\1)}$. Now, notice that assumption \eqref{wRi} implies that,  in the \tcb{action profile} $x^*$ defined above, every player $i$ in ${\mc R}$ is playing a strict best response $\mc B_i(x^*_{-i})=\{a\}$. This, combined with the BR-invariance of $\ov{\mc N}_{\mc S}^{(a\1)}$ for the $\mc S$-restricted {\black SNG} with\tcb{action profile} of players in $\mc R$ frozen to $y=a\1$, implies that 
	$$\ov{\mc N}=\{x^*\in\mc X:\,x^*_{\mc R}\in\{\pm\1\}\,,\, x^*_{\mc S}\in\bar{\mc N}_{{\mc S}}^{(x^*_{\mc R})}\}\,,$$
	is a non-empty, globally BR-stable subset of Nash equilibria, thus proving point (ii) of the claim. 
	\end{IEEEproof}

\begin{remark} In the special case of a network coordination game, i.e., when the network $\mc G$ is unsigned, Proposition \ref{prop:stability} provides sufficient conditions for consensus \tcb{action profile}s $\pm\1$ to be Nash equilibria and form a globally BR-reachable set. In this setting, the assumptions reduce to the following two conditions: (a) that $\mc G_{\mc R}$ is $(h^-,h^+)$-indecomposable; and (b) that $w_i>|h_i|$ for every $i$ in $\mc V$. Condition (a) ensures that no coexistent Nash equilibrium exists, while condition (b) ensures that both consensus \tcb{action profile}s $\pm\1$ are strict Nash equilibria. We notice that in this case conditions (a) and (b) are not just sufficient but also necessary for the set $\{\pm\1\}$ of consensus \tcb{action profile}s to be globally BR-stable. 
\end{remark}

We now present the following result extending Proposition \ref{prop:stability} to cases when $\mc G_{\mc R}$ is structurally balanced rather than unsigned. 

\begin{theorem}\label{theo:main-stability}
	Consider a {\black SNG} with binary actions on a network $\mc G=(\mc V,\mc E,W)$ with external field $h$. 
	 Consider a binary partition as in \eqref{binary-partition}, such that $\mc G_\mc R$ is structurally balanced 
	 and let $\tau$ in $\mc X_{\mc R}$ be such that $\mc G_{\mc R}^{[\tau]}$ is unsigned. 
	Let $h^-$ and $h^+$ in $\R^{\mc R}$ be the vectors with entries 
\be\label{h+-tau}h^+_i=\tau_ih_i+w_i^{{\mc S}}\,, \qquad h^-_i=\tau_ih_i-w_i^{{\mc S}}\,,\qquad\forall i\in\mc R\,.\ee
Assume that $\mc G_{\mc R}^{[\tau]}$ is $(h^-,h^+)$-indecomposable and that \eqref{wRi} holds true. 
Then:  
\begin{enumerate}
\item[(i)] if, for $y=\pm\tau$ in $\mc X_{\mc R}$, the set 
$\mc N_{{\mc S}}^{(y)}$  is globally BR-reachable for the $\mc S$-restricted {\black SNG} with the \tcb{action profile} in $\mc R$ frozen to $y$, 
then, the subset of Nash equilibria \be\label{Ntau}\ov {\mc N}=\{x^*\in\mc N:x^*_{\mc R}\in\{\pm\tau\}\}\,,\ee  is globally BR-reachable. 
	\end{enumerate}
	Moreover: 
	\begin{enumerate}
	\item[(ii)] if, for $y=\pm\tau$ in $\mc X_{\mc R}$, there exists a non-empty subset $\bar{\mc N}_{{\mc S}}^{(y)}\subseteq \mc N_{{\mc S}}^{(y)}$ that is globally BR-stable for the $\mc S$-restricted {\black SNG} where actions of players in $\mc R$ are frozen to $y$, then 
	there exists a globally BR-stable subset of Nash equilibria contained in $\ov{\mc N}$. 
\end{enumerate}
\end{theorem}
\begin{IEEEproof}
Let $\sigma$ in $\mc X$ be such that $\sigma_{\mc R}=\tau$ and $\sigma_{\mc S}=\1$. 
Consider the  associated gauge transformation $[\sigma]$ and the $[\sigma]$-transformed network  $\mc G^{[\sigma]}$. Observe that the assumptions ensure that $(\mc G^{[\sigma]})_{\mc R}=\mc G_{\mc R}^{[\tau]}$ is unsigned, that $\mc G_{\mc R}$ is $(h^-,h^+)$-indecomposable and that \eqref{wRi} holds true. We can thus apply Proposition \eqref{prop:stability} to the $[\sigma]$-transformed {\black SNG}, whose points (i) and (ii) imply, respectively, points (i) and (ii) of the claim. 
\end{IEEEproof}\medskip

Similarly to Section \ref{sec:existence}, we may derive the following two corollaries from Theorem \ref{theo:main-stability}.

\begin{corollary}\label{coro:reach} 
Consider a {\black SNG} on a network $\mc G=(\mc V,\mc E,W)$ with external field $h$. 
Consider a binary partition as in \eqref{binary-partition} such that $\mc G_{\mc R}$ and $\mc G_{\mc S}$ are both structurally balanced, and let $\tau$ in $\mc X_{\mc R}$ be such that $\mc G_{\mc R}^{[\tau]}$ is unsigned. Let $h^-$ and $h^+$ in $\R^{\mc R}$ be the vectors with entries as in \eqref{h+-}.
If $\mc G_{\mc R}$ is $(h^-,h^+)$-indecomposable and \eqref{wRi} holds true, then the subset of Nash equilibria \eqref{Ntau} is globally BR-reachable.
\end{corollary}
\begin{IEEEproof} If the subnetwork $\mc G_{\mc S}$ is structurally balanced, then, for every $\tau$ in $\mc X_{\mc R}$, Proposition \ref{prop:structurally-balanced-equilibria} implies that the set  $\mc N_{{\mc S}}^{(y)}$  is globally BR-reachable for the $\mc S$-restricted {\black SNG} with the \tcb{action profile} in $\mc R$ frozen to $y=\pm\tau$. 
The claim then follows from Theorem \ref{theo:main-stability}(i). 
\end{IEEEproof}
	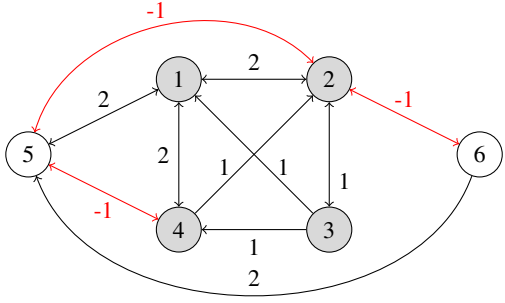
\begin{figure}
	\centering
	\begin{tikzpicture}
		\foreach \x/\name/\value in {(0,0)/4/1, (2,2)/2/1, (2,0)/3/1, (0,2)/1/1}\node[shape=circle,draw, fill=gray!30!white](\name) at \x {\small \name};
		\path [<->,draw] (2) edge node[right, near end] {\small 1} (3);
		\path [<->,draw] (2) edge node[above] {\small 2} (1);
		\path [<->,draw] (1) edge node[left] {\small 2} (4);
		\path [->,draw] (3) edge node[below] {\small 1} (4);
		\path [->,draw] (3) edge node[above, near start] {\small 1} (1);
		\path [->,draw] (4) edge node[above, near start] {\small 1} (2);

		\node[shape=circle, draw ](6) at (4,1) {\small 6};
		\node[shape=circle, draw ](5) at (-2,1) {\small 5};
		
		\path[<-, draw, bend right=70] (5) edge node[above] {\small 2} (6);
		\path[<->, draw, bend right=60,red] (2) edge node[above] {\small -1} (5);
		
		\path[<->, draw,red] (2) edge node[above] {\small -1} (6);
		\path[<->, draw] (1) edge node[above] {\small 2} (5);
		\path[<->, draw, red] (4) edge node[below] {\small -1} (5);
	\end{tikzpicture}	
	\caption{Signed graph with coordinating set $\mc R=\{1, \dots, 4\}$ in gray (see Example \ref{ex:non_dec_sup}).}
	\label{fig:non_dec_sup}
\end{figure}
\begin{example}\label{ex:non_dec_sup}
 Consider a {\black SNG} on the graph in Figure \ref{fig:non_dec_sup} with $h=0$. Observe that  $\mc R=\{1,\dots, 4\}$ and $\mc S=\mc V\setminus \mc R$ are  coordinating sets with $w^{\mc R}\vert_{\mc R}-\abs{d}_{\mc R}=(4,3,3,3)> w^{\mc S}\vert_{\mc R}=(2,2,0,1)$. 
Furthermore, the graph $\mc G_{\mc R}$ is $(h^+, h^-)$-indecomposable for $h^{+}=h_{\mc R} +w^{\mc S}\vert_{\mc R}=(2,2,0,1)$ and $h^-=h_{\mc R} -w^{\mc S}\vert_{\mc R}=-h^+$. This can be proved following the same reasoning as in Example \ref{ex:dec1} (notice that $\mc G_{\mc R}$ coincides with the graph in Figure \ref{fig:non_dec}). 
Then, according to Corollary \ref{coro:reach}, the set of Nash equilibria where the players in ${\mc R}$ are at consensus is globally BR-reachable.
Notice that the set of Nash equilibria 
$$\mc N=\{(a\1_{\mc R},-a,a),(a\1_{\mc R},a,-a) , a = \pm 1\}\,.$$ is not globally BR-stable. Indeed, for every $x^*$ in $\mc N$, the best response of player $5$ is $\mc B_5(x_{-5}^*)=\{\pm 1\}$, while $6$ is playing a strict best response. Therefore, there exists a best-response path (e.g., $((a\1_{\mc R},-a,a), (a\1_{\mc R},a,a))$ for $a=\pm 1$) that leaves the set of Nash equilibria. 
\end{example}
\begin{corollary}\label{coro:stability} 
Consider a {\black SNG} on a network $\mc G=(\mc V,\mc E,W)$ with external field $h$. 
Consider a binary partition as in \eqref{binary-partition} such that $\mc G_{\mc R}$ is structurally balanced and $\mc G_{\mc S}$ is undirected. Let $\tau$ in $\mc X_{\mc R}$ be such that $\mc G_{\mc R}^{[\tau]}$ is unsigned and let $h^-$ and $h^+$ in $\R^{\mc R}$ be the vectors with entries as in \eqref{h+-}.
If $\mc G_{\mc R}$ is $(h^-,h^+)$-indecomposable and \eqref{wRi} holds true, then there exists a globally BR-stable subset of Nash equilibria.
\end{corollary}
\begin{IEEEproof}
	If the subnetwork $\mc G_{\mc S}$ is undirected,  then Proposition \ref{pr:pot_mixed} implies that, for $y=\pm\tau$ in $\mc X_{\mc R}$, there exists a non-empty subset $\bar{\mc N}_{{\mc S}}^{(y)}\subseteq \mc N_{{\mc S}}^{(y)}$ that is globally BR-stable for the $\mc S$-restricted {\black SNG} where actions of players in $\mc R$ are frozen to $y$.
The claim then follows from Theorem \ref{theo:main-stability}(ii). 
\end{IEEEproof}

{\color{black}
	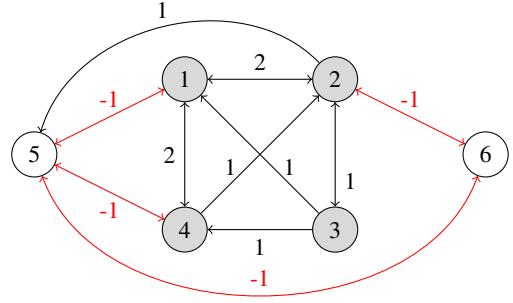
\begin{figure}
		\centering
		\begin{tikzpicture}
			\foreach \x/\name/\value in {(0,0)/4/1, (2,2)/2/1, (2,0)/3/1, (0,2)/1/1}\node[shape=circle,draw, fill=gray!30!white](\name) at \x {\small \name};
			\path [<->,draw] (2) edge node[right, near end] {\small 1} (3);
			\path [<->,draw] (2) edge node[above] {\small 2} (1);
			\path [<->,draw] (1) edge node[left] {\small 2} (4);
			\path [->,draw] (3) edge node[below] {\small 1} (4);
			\path [->,draw] (3) edge node[above, near start] {\small 1} (1);
			\path [->,draw] (4) edge node[above, near start] {\small 1} (2);

			\node[shape=circle, draw ](6) at (4,1) {\small 6};
			\node[shape=circle, draw ](5) at (-2,1) {\small 5};
			
			\path[<->, draw, bend right=70, red] (5) edge node[above] {\small -1} (6);
			\path[->, draw, bend right=60] (2) edge node[above] {\small 1} (5);
			
			\path[<->, draw,red] (2) edge node[above] {\small -1} (6);
			\path[<->, draw, red] (1) edge node[above] {\small -1} (5);
			\path[<->, draw, red] (4) edge node[below] {\small -1} (5);
		\end{tikzpicture}	
		\caption{Graph studied in Example \ref{ex:non_dec}.}
		\label{fig:non_dec}
	\end{figure}
\begin{example}\label{ex:non_dec}
Consider the {\black SNG} on the graph in Figure \ref{fig:non_dec} with $h=(2,0,0,-1,0,0)$. Observe that $\mc R=\{1,\dots, 4\}$ is a coordinating set with $w^{\mc R}\vert_{\mc R}=(4,3,3,3)$ and  $w^{\mc S}\vert_{\mc R}=(1,2,0,1)$ and $\mc G_{\mc S}$ with $\mc S=\mc V\setminus \mc R$ is undirected.
It holds that 
$$
w^{\mc R}\vert_{\mc R}-\abs{h}_{\mc R}=(2,3,3,2)>w^{\mc S}\vert_{\mc R}= (1,2,0,1)\,.
$$
Furthermore, the graph $\mc G_\mc R$ is $(h^{-},h^{+})$-indecomposable for $\tilde h^{+}=h\vert_{\mc R}+w^{\mc S}\vert_{\mc R}=(3,2,0,0)$ and $\tilde h^-=h\vert_{\mc R}-w^{\mc S}\vert_{\mc R}=(1,-2,0,-2)$. Again, this can proved following the same reasoning as in Example \ref{ex:dec1} and \ref{ex:non_dec_sup}. Then,  according to Corollary \ref{coro:stability}, the set of Nash equilibria where the players in ${\mc R}$ are at consensus contains a globally BR-stable subset.
\end{example}
\begin{example}\label{ex:dec}
Consider now the {\black SNG} on the graph in Figure \ref{fig:dec} and let $h=0$. Analogously to Example \ref{ex:non_dec}, we have that $\mc R=\{1,\dots, 4\}$ is a coordinating set with $w^{\mc R}\vert_{\mc R}=(4,3,3,3)$, 
 and $\mc G_{\mc S}$ with $\mc S=\mc V\setminus \mc R$ is undirected. On the other hand, in this case, we have $w^{\mc S}\vert_{\mc R}=(1,1,1,1)$. 
Since $w^{\mc R}\vert_{\mc R}\geq w^{\mc S}\vert_{\mc R}$, Corollary \ref{th:mixed_existence2} holds and existence of a Nash equilibrium where the players in $\mc R$ are at consensus is guaranteed. However, 
$\mc G_{\mc R}$ is \textit{not} $(h^+,h^-)$-indecomposable for $h^{+}=h+w^{\mc S}\vert_{\mc R}=(1,1,1,1)$ and $h^-=h-w^{\mc S}\vert_{\mc R}=-h^+$. Indecomposability is indeed violated by $\mc R^{+}=\{1,4\}$ and $\mc R^-=\{2,3\}$. Therefore, Proposition \ref{prop:stability} and Theorem \ref{theo:main-stability} do not apply. Observe that $x^*=(1,-1,-1,1, -1,1, 1,-1)$ is a strict Nash equilibrium of the game where players in $\mc R$ are not at consensus. 
\end{example}
}

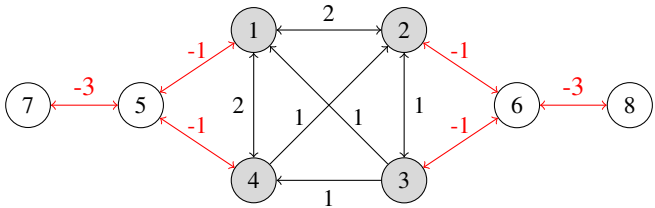
\begin{figure}
		\centering
		\begin{tikzpicture}
	\foreach \x/\name/\value in {(0,0)/4/1, (2,2)/2/1, (2,0)/3/1, (0,2)/1/1}\node[shape=circle,draw, , fill=gray!30!white](\name) at \x {\small \name};
	\path [<->,draw] (2) edge node[right] {\small 1} (3);
	\path [<->,draw] (2) edge node[above] {\small 2} (1);
	\path [<->,draw] (1) edge node[left] {\small 2} (4);
	\path [->,draw] (3) edge node[below] {\small 1} (4);
	\path [->,draw] (3) edge node[above, near start] {\small 1} (1);
	\path [->,draw] (4) edge node[above, near start] {\small 1} (2);
	
	\node[shape=circle, draw ](5) at (-1.5,1) {\small 5};
	\node[shape=circle, draw ](6) at (3.5,1) {\small 6};
	
	\path[<->, draw,red] (4) edge node[above] {\small -1} (5);
	\path[<->, draw, red] (1) edge node[above] {\small -1} (5);
	\path[<->, draw,red] (3) edge node[above] {\small -1} (6);
	\path[<->, draw,red] (2) edge node[above] {\small -1} (6);
	
	\node[shape=circle, draw ](7) at (-3,1){\small 7};
	\node[shape=circle, draw ](8) at (5,1){\small 8};
	
	\path [<->,draw=red, text=red] (7) edge node[above] {-3} (5);
	\path [<->,draw=red, text=red] (8) edge node[above] {-3} (6);


\end{tikzpicture}	
\caption{Graph studied in Example \ref{ex:dec}}
\label{fig:dec}
\end{figure}

\addtocounter{example}{-5}
\begin{example}[cont'd]
	Consider the graph in Figure \ref{fig:sb}. As observed in Example \ref{ex:sb} subset $\mc R=\{1, \dots, 4\}$ is such that the graph $\mc G_{\mc R}$ is structurally balanced. Indeed, if we consider the gauge transformation $[\sigma]$ with $\sigma=[1,-1,-1,1,1,1]$, we obtain that $W^{[\sigma]}=[\sigma]W[\sigma]\geq 0$. Furthermore, for $h_i=0$, we have $w^{\mc R}\vert_{\mc R}=(4,3, 4, 3)> w^{\mc S}\vert_{\mc R}=(1,1,1,1) $. Since $\mc G_{\mc S}$ is a undirected (no links), this implies that, for $\tau=(1,-1,-1,1)$, there exists a polarized Nash equilibrium, that is, a Nash equilibrium where $x_i^*=\tau_i$ for all $i$ in $\mc R$. For instance, the \tcb{action profile} $x^*=(1,-1,-1,1,-1,1)$ is a Nash equilibrium of the game.

Finally, since $\tilde{\mc G}_{\mc R}$ with $\tilde W_{\mc R\mc R}=\abs{W_{\mc R\mc R}}$ is $(h^+,h^-)$-indecomposable for $h^+=\1$ and $h^-=-\1$, the set of Nash equilibria that are polarized in the structurally balanced component admits a globally BR-stable subset. We remark that these results can be obtained \textit{without } knowing the partitions $\mc R_1$ and $\mc R_2$. 
\end{example}

\section{Concluding remarks}\label{conclusion}
The {\black signed network games with binary actions} studied in this paper encompass a number of network strategic games that have appeared in the recent literature. They model the contemporaneous presence of strategic complement and strategic substitute effects in an economic multi-player model, or rather the presence of antagonistic behaviors in a social network.
Such games pose significant challenges: Nash equilibria may fail to exist, and even when they do, learning dynamics may be complex and sensitive to initial conditions and activation order. In this paper, we have obtained conditions under which a subset of coordinating players is capable of forcing the convergence of best response dynamics to a Nash equilibrium that is a consensus on their part. Our results use in a novel way the concept of cohesiveness proposed in \cite{Morris:2000} and build on the super-modular properties of coordinating games.  Further work includes finding efficient algorithms to verify the proposed conditions and deriving necessary conditions for existence, reachability and stability of Nash equilibria.

\section*{Acknowledgments}
The authors gratefully acknowledge Prof.~Claudio Altafini and Prof.~Angela Fontan for insightful discussions.
\bibliographystyle{ieeetr}
\bibliography{bib}

\begin{thebibliography}{10}

\bibitem{vanelli2020games}
M.~Vanelli, L.~Arditti, G.~Como, and F.~Fagnani, ``On games with coordinating
  and anti-coordinating agents,'' {\em IFAC-PapersOnLine}, vol.~53, no.~2,
  pp.~10975--10980, 2020.

\bibitem{arditti2021equilibria}
L.~Arditti, G.~Como, F.~Fagnani, and M.~Vanelli, ``Equilibria and learning
  dynamics in mixed network coordination/anti-coordination games,'' in {\em
  2021 60th IEEE Conference on Decision and Control (CDC)}, pp.~4982--4987,
  2021.

\bibitem{Blume:1993}
L.~Blume, ``The statistical mechanics of strategic interaction,'' {\em Games
  and Economic Behavior}, vol.~5, pp.~387--424, 1993.

\bibitem{Jackson:2008}
M.~O. Jackson, {\em Social and Economic Networks}.
\newblock Princeton University Press, 2008.

\bibitem{Galeotti.ea:2010}
A.~Galeotti, S.~Goyal, M.~Jackson, F.~Vega-Redondo, and L.~Yariv, ``Network
  games,'' {\em The Review of Economic Studies}, vol.~77, no.~1, pp.~218--244,
  2010.

\bibitem{Jackson.Zenou:2015}
M.~O. Jackson and Y.~Zenou, {\em Handbook of game theory with economic
  applications}, vol.~4, ch.~Games on networks, pp.~95--163.
\newblock Elsevier, 2015.

\bibitem{Ellison:1993}
G.~Ellison, ``Learning, local interaction, and coordination,'' {\em
  Econometrica}, vol.~61, no.~5, pp.~1047--1071, 1993.

\bibitem{Young:1993}
H.~P. Young, ``The evolution of conventions,'' {\em Econometrica: Journal of
  the Econometric Society}, vol.~61, no.~1, pp.~57--84, 1993.

\bibitem{Morris:2000}
S.~Morris, ``Contagion,'' {\em The Review of Economic Studies}, vol.~67, no.~1,
  pp.~57--78, 2000.

\bibitem{Young:2006}
H.~Young, {\em The Diffusion of Innovations in Social Networks}, vol.~Economy
  as an evolving complex system, pp.~267--282.
\newblock Oxford University Press US, 2006.

\bibitem{Jackson.Storms:2025}
M.~Jackson and E.~Storms, ``Behavioral communities and the atomic structure of
  networks,'' {\em American Economic Journal: Microeconomics}, vol.~18, no.~1,
  pp.~146--173, 2026.

\bibitem{Bramoulle.ea:2004}
Y.~Bramoull\'e, D.~L\'opez-Pintado, S.~Goyal, and F.~Vega-Redondo, ``Network
  formation and anti-coordination games,'' {\em International Journal of Game
  Theory}, vol.~33, pp.~1--19, 2004.

\bibitem{Galam:2004}
S.~Galam, ``Contrarian deterministic effects on opinion dynamics: the hung
  elections scenario,'' {\em Phisica A}, vol.~333, no.~53--460, 200.

\bibitem{Bramoulle:2007}
Y.~Bramoull{\'e}, ``Anti-coordination and social interactions,'' {\em Games and
  Economic Behavior}, vol.~58, no.~1, pp.~30--49, 2007.

\bibitem{Lopez-Pintado:2009}
D.~L\'opez-Pintado, ``Network formation, cost-sharing and anti-coordination,''
  {\em International Game Theory Review}, vol.~11, no.~1, pp.~53--76, 2009.

\bibitem{Grabisch.Li:2019}
M.~Grabisch and F.~Li, ``Anti-conformism in the threshold model of collective
  behavior,'' {\em Dynamic Games and Applications}, vol.~9, pp.~1--34, 2019.

\bibitem{Topkins:1979}
D.~M. Topkins, ``Equilibrium points in nonzero-sum n-person submodular games,''
  {\em SIAM Journal on Control and Optimization}, vol.~17, no.~6, pp.~773--787,
  1979.

\bibitem{Milgrom.ea:1990}
P.~Milgrom and J.~Roberts., ``Rationalizability, learning, and equilibrium in
  games with strategic complementarities,'' {\em Econometrica}, vol.~58,
  pp.~1255--1277, 1990.

\bibitem{Vives:1990}
X.~Vives, ``Nash equilibrium with strategic complementarities,'' {\em Journal
  of Mathematical Economics}, vol.~19, pp.~305--321, 1990.

\bibitem{Topkins:1998}
D.~M. Topkins, {\em Supermodularity and Complementarity}.
\newblock Princeton University Press, 1998.

\bibitem{Ramazi.Cao:2020}
P.~Ramazi and M.~Cao, ``Convergence of linear threshold decisionmaking dynamics
  in finite heterogeneous populations,'' {\em Automatica}, vol.~119, p.~109063,
  2020.

\bibitem{sakhaei2023equilibration}
N.~Sakhaei, Z.~Maleki, and P.~Ramazi, ``Equilibration analysis and control of
  coordinating decision-making populations,'' {\em IEEE Transactions on
  Automatic Control}, vol.~69, no.~8, pp.~5065--5080, 2023.

\bibitem{Ramazi.Riehl.Cao:2016}
P.~Ramazi, J.~Riehl, and M.~Cao, ``Networks of conforming and nonconforming
  individuals tend to reach satisfactory decisions,'' {\em Proceedings of the
  National Academy of Sciences of the United States of America}, vol.~113,
  no.~46, pp.~12985--12990, 2016.

\bibitem{Grabisch.ea:2019}
M.~Grabisch, A.~Poindron, and A.~Ruzinowska, ``A model of anonymous influence
  with anti-conformist agents,'' {\em Journal of Economic Dynamics \& Control},
  vol.~109, p.~103773, 2019.

\bibitem{ramazi2023characterizing}
P.~Ramazi and M.~H. Roohi, ``Characterizing oscillations in heterogeneous
  populations of coordinators and anticoordinators,'' {\em Automatica},
  vol.~154, p.~111068, 2023.

\bibitem{le2023heterogeneous}
H.~Le, A.~Aghaeeyan, and P.~Ramazi, ``Heterogeneous mixed populations of
  conformists, nonconformists, and imitators,'' {\em IEEE Transactions on
  Automatic Control}, vol.~69, no.~5, pp.~3373--3380, 2023.

\bibitem{Monaco:2016}
A.~Monaco and T.~Sabarwal, ``Games with strategic complements and
  substitutes,'' {\em Economic Theory}, vol.~62, no.~1, pp.~65--91, 2016.

\bibitem{Harary:1953}
F.~Harary, ``On the notion of balance of a signed graph,'' {\em Michigan Math.
  J.}, vol.~2, pp.~143--146, 1953.

\bibitem{Cartwright:1956}
D.~Cartwright and F.~Harary, ``Structural balance: A generalization of heider's
  theory,'' {\em The Psychological Review}, vol.~63, no.~5, pp.~277--293, 1956.

\bibitem{Macy:2003}
M.~W. Macy, J.~A. Kitts, A.~Flache, and S.~Benard, ``{Polarization in Dynamic
  Networks: A Hopfield Model of Emergent Structure},'' in {\em Dynamic Social
  Network Modeling and Analysis}, pp.~162--173, The National Academies Press,
  2003.

\bibitem{Leskovec:2010}
J.~Leskovec, D.~Huttenlocker, and J.~Kleinberg, ``Signed networks in social
  media,'' in {\em CHI '10: Proceedings of the SIGCHI Conference on Human
  Factors in Computing Systems}, pp.~1361--1370, 2010.

\bibitem{Kauffman:1969}
S.~A. Kauffman, ``Metabolic stability and epigenesis in randomly constructed
  genetic nets,'' {\em Journal of Theoretical Biology}, vol.~22, no.~3,
  pp.~437--467, 1969.

\bibitem{Thomas:1973}
R.~Thomas, ``Boolean formalization of genetic control circuits,'' {\em Journal
  of Theoretical Biology}, vol.~42, pp.~563--585, 1973.

\bibitem{Hopfield:2010}
J.~J. Hopfield, ``Neural networks and physical systems with emergent collective
  computational abilities,'' {\em Proceedings of the National Academy of
  Sciences}, vol.~79, no.~8, pp.~2554--2558, 1982.

\bibitem{Mezard}
M.~Mezard, G.~Parisi, and M.~A. Virasoro, ``Spin glass theory and beyond,''
  {\em World Scientific Publishing Company}, 1987.

\bibitem{Granovetter:1978}
M.~Granovetter, ``Threshold models of collective behavior,'' {\em American
  Journal of Sociology}, vol.~83, no.~6, pp.~1420--1443, 1978.

\bibitem{He:2013}
X.~He, H.~Du, M.~W. Feldman, and G.~Li, ``Information diffusion in signed
  networks,'' {\em PLoS ONE}, vol.~14, no.~10, pp.~1--21, 2019.

\bibitem{Golesa:2025}
E.~Golesa, P.~Montealegrea, R.~Rios-Wilson, and S.~Sen{\'e}, ``Dynamical
  stability of threshold networks over undirected signed graphs,'' {\em
  Theoretical Computer Science}, vol.~1042, p.~115229, 2025.

\bibitem{Altafini:2012}
C.~Altafini, ``Dynamics of opinion forming in structurally balanced social
  networks,'' {\em PLoS ONE}, vol.~7, no.~6, pp.~935--946, 2012.

\bibitem{Altafini:2013}
C.~Altafini, ``Consensus problems on networks with antagonistic interactions,''
  {\em IEEE Transactions on Automatic Control}, vol.~58, no.~4, pp.~935--946,
  2013.

\bibitem{fontan2017multiequilibria}
A.~Fontan and C.~Altafini, ``Multiequilibria analysis for a class of collective
  decision-making networked systems,'' {\em IEEE Transactions on Control of
  Network Systems}, vol.~5, no.~4, pp.~1931--1940, 2017.

\bibitem{fontan2021role}
A.~Fontan and C.~Altafini, ``The role of frustration in collective
  decision-making dynamical processes on multiagent signed networks,'' {\em
  IEEE Transactions on Automatic Control}, vol.~67, no.~10, pp.~5191--5206,
  2021.

\bibitem{arditti2024robust}
L.~Arditti, G.~Como, F.~Fagnani, and M.~Vanelli, ``Robust coordination of
  linear threshold dynamics on directed weighted networks,'' {\em IEEE
  Transactions on Automatic Control}, vol.~69, no.~10, pp.~6515--6529, 2024.

\bibitem{Yanovskaya:1968}
E.~B. Yanovskaya, ``Equilibrium points in polymatrix games,'' {\em Litovskii
  Matematicheskii Sbornik}, vol.~8, pp.~381--384, 1968.

\bibitem{Arditti.Como.Fagnani:2024}
L.~Arditti, G.~Como, and F.~Fagnani, ``On the separability of functions and
  games,'' {\em IEEE Transactions on Control of Network Systems}, vol.~11,
  no.~2, pp.~831--841, 2024.

\bibitem{Blume:1995}
L.~Blume, ``The statistical mechanics of best response strategy revision,''
  {\em Games and Economic Behavior}, vol.~11, no.~2, pp.~111--145, 1995.

\bibitem{Monderer.Shapley:1996}
D.~Monderer and L.~Shapley, ``Potential games,'' {\em Games and Economic
  Behavior}, vol.~14, pp.~124--143, 1996.

\bibitem{Catalano.ea:2024}
C.~Catalano, M.~Castaldo, G.~Como, and F.~Fagnani, ``On a network centrality
  maximization game,'' {\em Mathematics of Operations Research}, vol.~50,
  no.~3, pp.~2112--2140, 2025.

\end{thebibliography}

\appendices
\section{Proof of Lemma \ref{lemma:NashGauge}}\label{sec:proof-lemma-NashGauge}
First, observe that 
$x^*$ in $\mc X$ is a Nash equilibrium for the $[\sigma]$-transformed game if and only if 
$$x_i^*\left(\sum\nolimits_jW^{[\sigma]}_{ij}x^*_j+h^{[\sigma]}_i\right)\ge0\,,\qquad \forall i\in\mc V\,.$$
Now, notice that 
$$
\ba{rcl}
\!\!\!\ds x_i^*\left(\sum\nolimits_jW^{[\sigma]}_{ij}x^*_j+h^{[\sigma]}_i\right)
&\!\!\!\!=\!\!\!\!&
x_i^*\left(\ds\sum\nolimits_j\sigma_iW_{ij}\sigma_jx^*_j+\sigma_ih_i\right)\\[10pt]
&\!\!\!\!=\!\!\!\!&\ds
\sigma_ix_i^*\left(\sum\nolimits_jW_{ij}\sigma_jx^*_j+h_i\right)\\[10pt]
&\!\!\!\!=\!\!\!\!&\ds
([\sigma]x^*)_i\left(\sum\nolimits_jW_{ij}([\sigma]x^*)_j+h_i\right)
.\ea$$
Therefore, $x^*$ is a Nash equilibrium for the $[\sigma]$-transformed game if and only if  $[\sigma]x^*$ is a Nash equilibrium for the 
{\black SNG} with binary actions on $\mc G$ with external field $h$. $\qed$

\section{Proof of Lemma \ref{lemma:balanced}}\label{sec:proof-lemma-balanced}
If $\mc G$ is structurally balanced, then  consider a balanced partition as in \eqref{balanced-partition} and let $\sigma$ in $\mc X$ have entries $\sigma_i=-1$ for every $i$ in $\mc V_1$ and $\sigma_i=+1$ for every $i$ in $\mc V_2$. It then  follows from \eqref{balanced-partition-1} that 
$W^{[\sigma]}_{ij}=\sigma_iW_{ij}\sigma_j=W_{ij}\geq 0$, 
for every $i$ and $j$ in $\mc V_q$, for $q=1,2$, whereas \eqref{balanced-partition-2} implies that 
$W^{[\sigma]}_{ij}=\sigma_iW_{ij}\sigma_j=-W_{ij}\geq 0$, for every $i$ in $\mc V_q$ and $j$ in $\mc V_r$, for $q\ne r$, $q,r=1,2$. This shows that, if $\mc G$ is structurally balanced, then $W^{[\sigma]}$ is a nonnegative matrix, hence the whole set $\mc V$ is coordinating for the transformed network $\mc G^{[\sigma]}$. 

Conversely, let $[\sigma]$ be a gauge transformation such that the whole node set $\mc V$ is coordinating for the transformed network $\mc G^{[\sigma]}$, i.e., $W^{[\sigma]}_{ij}\ge0$ for every $i$ and $j$ in $\mc V$. Define \be\label{V12}\mc V_1=\{i\in\mc V\!:\,\sigma_i=-1\}\,,\qquad\mc V_2=\{i\in\mc V\!:\,\sigma_i=1\}\,.\ee Then, clearly \eqref{balanced-partition} holds true. Moreover, for every $i$ and $j$ in $\mc V_q$, for $q=1,2$, we have  $W_{ij}=\sigma_iW^{[\sigma]}_{ij}\sigma_j=W_{ij}^{[\sigma]}\geq 0$, so that  \eqref{balanced-partition-1} holds true. Furthermore,  for every $i$ in $\mc V_q$ and $j$ in $\mc V_r$, for $q\ne r$, we have $W_{ij}=\sigma_iW^{[\sigma]}_{ij}\sigma_j=-W^{[\sigma]}_{ij}\geq 0$, so that  \eqref{balanced-partition-2} holds true as well. Therefore, \eqref{V12} determines a balanced partition of the node set $\mc V$, so that $\mc G$ is structurally balanced. \qed

\begin{IEEEbiography}
[{\includegraphics[width=1in,height=1.25in,clip,keepaspectratio]{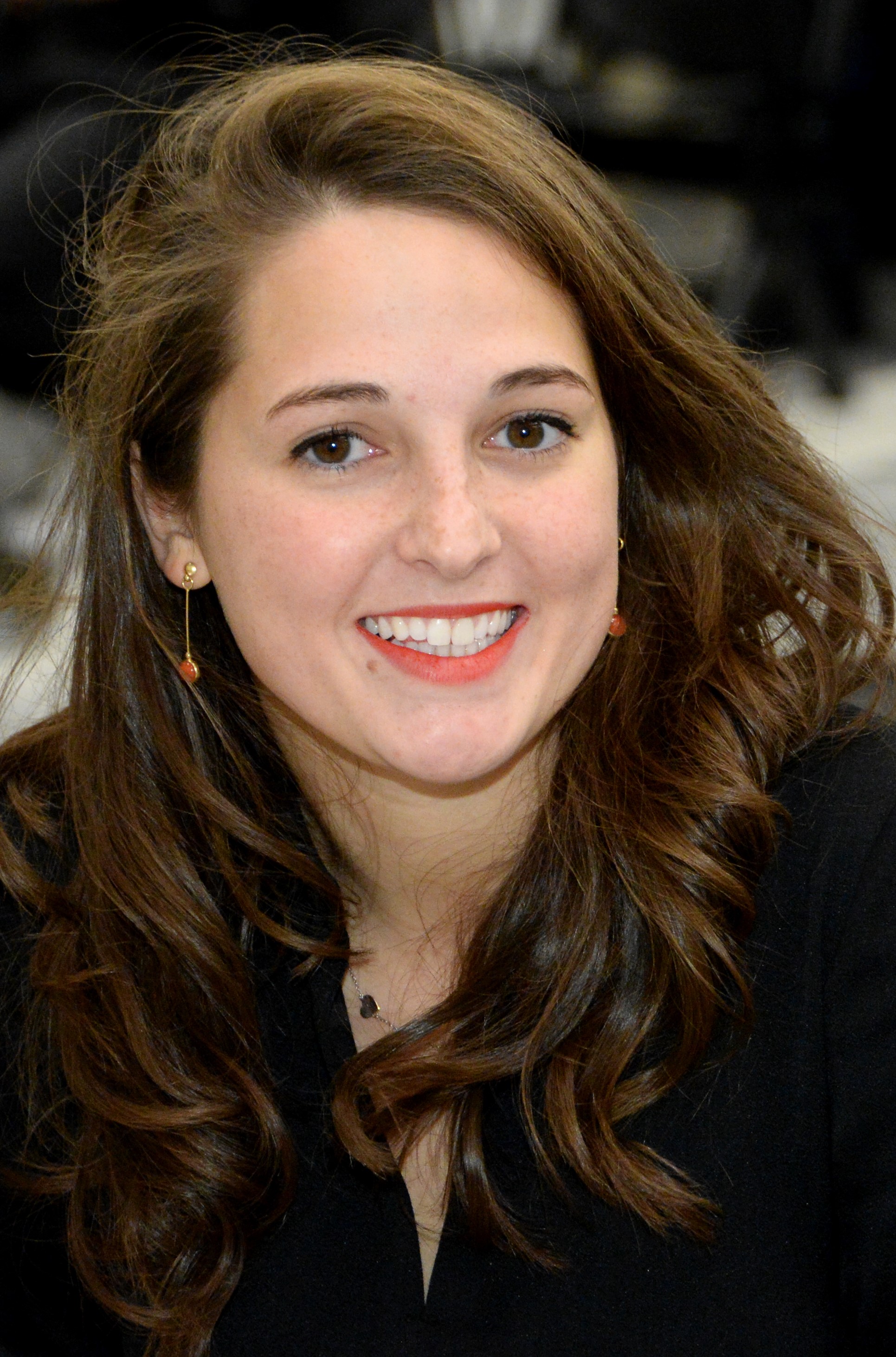}}]
{Martina Vanelli} received the B.Sc., M.S. (\textit{cum laude}), and Ph.D. (\textit{with honors}) degrees in Applied Mathematics  from  Politecnico  di  Torino,  in  2017,  2019, and 2024, respectively. She is currently a postdoctoral fellow at the Institute for Information and Communication Technologies, Electronics and Applied Mathematics (ICTEAM), Université catholique de Louvain. From October 2018 to March 2019, she was a visiting student at Technion, Israel Institute of Technology. Her research interests include identification, analysis, and control of multi-agent systems, with application to social and economic networks, and power markets. 
\end{IEEEbiography}

\begin{IEEEbiography}
[{\includegraphics[width=1in,height=1.25in,clip,keepaspectratio]{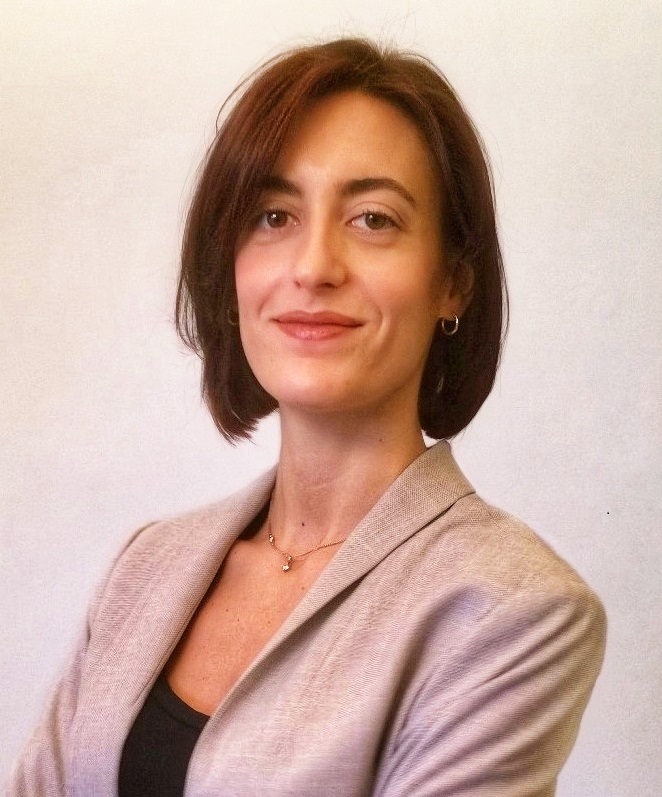}}]
{Laura Arditti} graduated with honors as a PhD in Applied Mathematics at the Department of Mathematical Sciences of Politecnico di Torino in 2023. She previously received the B.Sc. in Physics Engineering in 2016 and the M.S. in Mathematical Engineering in 2018, both magna cum laude from Politecnico di Torino. She currently works in the
financial sector as quantitative Financial Crime Prevention analyst. Her research focused on game theory, its relationship with graphical models, and its
applications to economic and financial networks.
\end{IEEEbiography}

\begin{IEEEbiography}
[{\includegraphics[width=1in,height=1.25in,clip,keepaspectratio]{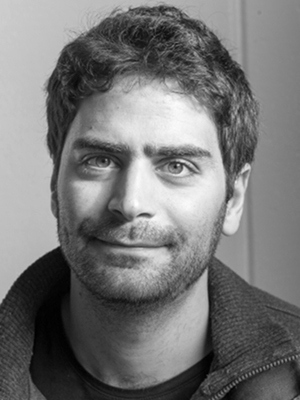}}]
{Giacomo Como}(M'12) {Giacomo Como} is  a  Professor at  the Department  of  Mathematical  Sciences,  Politecnico di  Torino,  Italy. He is also a Senior Lecturer  at  the  Automatic  Control  Department, Lund  University,  Sweden.  He  received the B.Sc., M.S., and Ph.D.~degrees in Applied Mathematics  from  Politecnico  di  Torino, Italy, in  2002,  2004, and 2008, respectively. He was a Visiting Assistant in  Research  at  Yale  University  in  2006--2007  and  a Postdoctoral  Associate  at  the  Laboratory  for  Information  and  Decision  Systems,  Massachusetts  Institute of Technology in  2008--2011. Prof.~Como currently serves as Senior Editor for the \textit{IEEE Transactions on Control of Network Systems}, and as Associate  Editor  for \textit{Automatica} and the \textit{IEEE Transactions on Automatic Control}.  He served as Associate Editor for the  \textit{IEEE Transactions on Network Science and Engineering} (2015-2021) and for the \textit{IEEE Transactions on Control of Network Systems} (2016-2022).  He was  the  IPC  chair  of  the  IFAC  Workshop  NecSys'15 ,  a  semiplenary speaker  at  the  International  Symposium  MTNS'16, and the  chair  of the  {IEEE-CSS  Technical  Committee  on  Networks  and  Communications} (2019-2024). He  is a  recipient  of  the 2015  George S.~Axelby  Outstanding Paper Award.  His  research interests  are in  dynamics,  information,  and  control  in  network  systems  with  applications to  cyber-physical  systems,  infrastructure  networks,  and  social  and  economic networks.
\end{IEEEbiography}

\begin{IEEEbiography}
[{\includegraphics[width=1in,height=1.25in,clip,keepaspectratio]{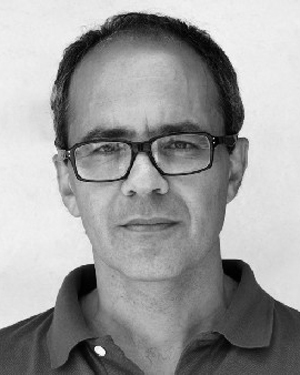}}]
{Fabio Fagnani} received the Laurea degree in Mathematics from the University of Pisa and the Scuola Normale Superiore, Pisa, Italy, in 1986. He received the PhD degree in Mathematics from the University of Groningen,  Groningen,  The  Netherlands,  in 1991. From 1991 to 1998, he was an Assistant Professor of Mathematical Analysis at the Scuola Normale Superiore. In 1997, he was a Visiting Professor at the Massachusetts Institute of Technology (MIT), Cambridge, MA. Since 1998, he has been with the Politecnico of Torino, where since 2002 he has been a Full Professor of Mathematical Analysis. From 2006 to 2012, he has acted as Coordinator of the PhD program in Mathematics for Engineering Sciences at Politecnico di Torino. From June 2012 to September 2019, he served as the Head of the Department of Mathematical Sciences, Politecnico di Torino. His current research topics are on cooperative algorithms and dynamical systems over networks, inferential distributed algorithms, and opinion dynamics. He is an Associate Editor of the \textit{IEEE Transactions on Automatic Control} and served in the same role for the \textit{IEEE Transactions on network Science and Engineering} and of the \textit{IEEE Transactions on Control of network Systems}.
\end{IEEEbiography}

\end{document}